\documentclass{llncs}
\usepackage{makeidx}  
\usepackage{amsmath}
\usepackage{graphicx}
\usepackage[caption = false]{subfig}
\usepackage{url}

\begin{document}
\begin{frontmatter}

\title{Passwords: Divided they Stand, United they Fall}
\author{First Author \and Second Author}
\author{Harshal Tupsamudre \and Sachin Lodha}
%
%
\institute{TCS Research, Pune, India\\
\email{harshal.tupsamudre@tcs.com, sachin.lodha@tcs.com}
}
\maketitle
\begin{abstract}
Today, offline attacks are one of the most severe threats to password security. These attacks have claimed millions of passwords from prominent websites including Yahoo, LinkedIn, Twitter, Sony, Adobe and many more. Therefore, as a preventive measure, it is necessary to gauge the offline guessing resistance of a password database and to help users choose secure passwords. The rule-based mechanisms that rely on minimum password length and different character classes are too naive to capture the intricate human behavior whereas those based on probabilistic models require the knowledge of an entire password distribution which is not always easy to learn.
\\

In this paper, we propose a space partition attack model which uses information from previous leaks, surveys, attacks and other sources to divide the password search space into non-overlapping partitions and learn partition densities. We prove that the expected success of a partition attacker is maximum if the resulting partitions are explored in decreasing order of density. We show that the proposed attack model is more general and various popular attack techniques including probabilistic attacker, dictionary-based attacker, grammar-based attacker and brute-force attacker are just different instances of a partition attacker. Later, we introduce bin attacker, another instance of a partition attacker, and measure the guessing resistance of real-world password databases. We demonstrate that the utilized search space is very small and as a result even a weak attacker can cause sufficient damage to the system. 
\\

We prove that partition attacks can be countered only if the partition densities are uniform. We use this result and propose a system that thwarts partition attacker by distributing users across different partitions. Finally, we demonstrate how some of the well-known password schemes can be adapted to help users in choosing passwords from the system assigned partitions and investigate their effectiveness by performing a usability study.
\end{abstract}


\end{frontmatter}

\section{Introduction}
Text-based passwords were introduced in the digital world in the mid-1960s to provide controlled access to time-sharing computers~\cite{history}. Today, with the proliferation of internet services such as banking, insurance, utility, media and many others, the role of text-based passwords has been renewed to protect access to the diverse range of resources residing on the remote web server. However, the use of human-generated text passwords as an authentication mechanism raises some serious concerns which need to be addressed. Text-based passwords are drawn from a relatively small space which makes them vulnerable to guessing attacks~\cite{Morris,Klein,Weir}. This vulnerability is due to the user's preference for common dictionary words and the predominant use of lowercase letters and digits. To achieve the desired level of security, various alternatives such as graphical passwords and biometrics have been proposed. However, the benefits provided by text passwords are unmatched and cannot be achieved by any of the existing alternatives~\cite{quest}. Consequently, text-based passwords still remain the dominant form of user authentication.

Depending upon whether the guessing is carried out remotely or locally, guessing attacks can be categorized into two types, online attack and offline attack. In an online attack, the attacker attempts to login into the system by guessing the password of a legitimate user. To improve the likelihood of a successful login, the online attacker sorts the password guesses in decreasing order of probability (learned from previous breaches) and tests popular passwords first. These attacks work well on websites with a large user base. For instance, the breach of around $32 \ million$ accounts from the Rockyou website revealed that the password ``123456'' was used by 290,731 ($\sim$1\%) of its users~\cite{Rockyou1}. Further, recently published password frequency list of $70  \ million$ Yahoo users~\cite{Yahoo} reveals the count of the most popular password to be 753,217 ($\sim$1\%). This real-world data suggests that the online attacker can compromise nearly 1\% of the  accounts by merely trying a single guess on the target website. The conventional lock-out strategy that limits the number of unsuccessful attempts does not counter such attacks which target many accounts with a handful of popular passwords. 

To mitigate the threat of online attacks, Schechter {\em et al.}~\cite{Schechter} proposed the use of a count-min sketch data-structure to count the occurrences of every password on a given system and to prohibit those passwords that reach a certain popularity threshold. Later, South {\em et al.}~\cite{South} proposed counting the occurrences of passwords as well as their variations that lie within hamming distance of two, especially targeting the leet transformations that replace letters `a' with `@', `i' with `!' and so on. Websites such as Twitter took a more direct approach by banning 370 popular passwords. Thus, blacklisting the use of popular passwords and their common variants can provide much better resistance to online attacks. 

The offline guessing attacks, on the other hand, pose a very serious threat to password security. In this scenario, the attacker has complete access to the database comprising passwords of all registered users of a website\footnote{We assume that these passwords are protected using a one-way hash function, otherwise no guessing is required.}. The attacker in possession of the password database can generate and verify potentially unlimited guesses against hashed passwords. And, if passwords are simple, such as dictionary words or small enough for brute-force search~\cite{Morris}, then the attacker can recover them in no time. 

Nowadays, the breach of a password database is a common event. In the past few years, millions of passwords have been stolen from prominent websites including Yahoo, LinkedIn, Hotmail, Twitter, Sony, Adobe and many others~\cite{pwned}. The continuous improvement in the state-of-the-art password crackers has resulted in a conglomeration of different attack techniques\footnote{While demonstrating the effectiveness of these attacks, researchers assumed that the target password database is protected using a one-way hash function with no salts.} ranging from simple dictionary attacks~\cite{Klein} to more sophisticated attacks based upon the Markov model~\cite{Narayanan} and probabilistic context-free grammar~\cite{Weir}. Moreover, with the advent of GPU computing, there has been a tremendous rise in the guessing capability of the offline attacker. For instance, in 2012, a 25-GPU cluster~\cite{GPU} was unveiled which could generate 350 billion NTLM hashes per second and therefore guess a standard eight character length Windows password ($95^{8}$ search space) in less than six hours. Further, the password guessing tools have already been optimized to work on GPU~\cite{hashcat}. {\em Thus, the frequent occurrences of high magnitude breaches along with the availability of sophisticated cracking techniques combined with the high speed GPU clusters has made offline password guessing a real threat, which needs to be addressed.}

Various attempts have been made in the past to measure the strength of a password database against offline attacks. The popular technique to infer the resistance of a password database is by simulating attacks on publicly available breached databases~\cite{Weir,Shay3,Chinese} or on passwords collected from surveys~\cite{Shay1,Shay2,Shay2014}. NIST provides an entropy measure for gauging password strength~\cite{NIST}, however, this approach is not based upon large empirical data and depends only upon the password length and character classes (such as lowercase, uppercase, digit and symbol) used for composing passwords. In 2012, Bonneau~\cite{BonneauMetric} proposed a partial guessing metric to measure the resistance of a password database against offline attacks. The attack model used in deriving this metric assumes the knowledge of a password distribution which is not easy to learn. This attack model accurately represents the all-powerful attacker who has complete information to enumerate all passwords in decreasing order of probability. However, there are circumstances where the attacker cannot be assumed to have this complete information. 

The offline attacker mainly exploits the fact that few passwords and their variations are very popular (these popular passwords are mostly learned from the breached databases such as Rockyou). However, a major fraction of the password database consists of infrequent passwords~\cite{Zipf} and their probabilities are not easy to learn. For instance, the most popular Rockyou password occurred 290,731 times, however, a long list (53\%) of Rockyou passwords occurred less than five times. Further, if a website implements new password policy that blacklists previously breached passwords and prohibits passwords exceeding certain frequency threshold (as proposed by Schechter {\em et al.}~\cite{Schechter}) then the attacker cannot estimate the probabilities of passwords in the target database until such type of passwords gets leaked.

In this paper, we propose a more general attack model which utilizes information from previous leaks, surveys, attacks and other sources to divide the entire password search space into non-overlapping partitions. Although the attacker cannot learn the password probabilities, we show that the attacker can still learn the partition level features and recover many passwords. The size of the resulting partitions depends upon the information possessed by the attacker. The more information regarding the target password the more granular are the partitions. 
We provide the intuition behind this partitioning attack model using an example. Suppose that the offline attacker acquires an unprotected password database comprising $\phi = 26$ passwords drawn from a small hypothetical search space $\sigma$. Further, suppose that after analysing these 26 plaintext passwords, the attacker divides the entire password search space into 16 non-overlapping chunks as shown in Figure~\ref{fig:partition}.
\begin{figure}[h]
\centering
\includegraphics[scale=0.6]{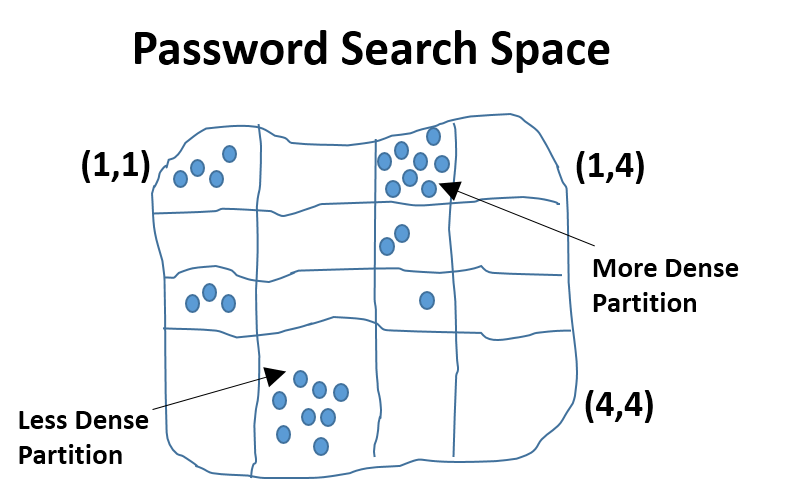}
\caption{Division of a hypothetical search space $\sigma$ into 16 partitions.}~\label{fig:partition}
\end{figure}
\begin{definition}
The division of a password search space $\sigma$ into $n$ non-overlapping chunks $\{\sigma_{1},\sigma_{2},\ldots,\sigma_{n}\}$ is known as partitioning and the chunks themselves are called as partitions.
\end{definition} 
As shown in Figure~\ref{fig:partition} all 26 passwords are distributed in just 6 partitions and the remaining 10 partitions are empty. Although the search space depicted in the figure is for illustration purpose, we shall see later in section~\ref{sec:binattack} that the situation of the real world passwords is very much similar. The fact that most partitions of the search space are never used clearly shows the gap between the available search space and the utilized search space. After learning the partitions using $\phi=26$ passwords, the attacker can now break the protected passwords in the target database by simply ignoring the 10 unused partitions and utilizing the available computing resources to explore these 6 utilized partitions. Thus, the search space for the attacker is reduced considerably.  However, we show that the attacker can do better than this.

For convenience, we view the password search space shown in Figure~\ref{fig:partition} as a 2-D array of 4X4 partitions and refer to its 16 partitions by a tuple $(x,y)$, where $x$ indicates the row and $y$ indicates the column number of the partition.  Now, consider the partitions $(1,1)$ and $(1,3)$ that have 4 and 8 passwords respectively. Thus, for creating passwords, the partition $(1,3)$ seems to be more popular than the partition $(1,1)$. We capture the notion of popularity by defining the probability of a partition. Let $\phi_{i}$ represent the number of passwords in a database that belong to $i^{th}$ partition. Therefore we have,
\begin{equation}
\phi = \sum\limits_{i=1}^{n}\phi_i
\end{equation}
\begin{definition}
The probability $p_{\sigma_i}$ of a partition $\sigma_{i}$ is the fraction of passwords in a password database that belong to the partition $\sigma_{i}$.
\begin{equation}
p_{\sigma_i} =  \phi_i/\phi
\end{equation}
\end{definition}
The probability of the partition $(1,1)$ is 4/26 while the probability of the partition $(1,3)$ is $8/26$. Thus, the partition $(1,3)$ is more popular than the partition $(1,1)$. Earlier, we saw that the attacker can benefit considerably by ignoring the partitions that are never used. But now, instead of exploring the utilized partitions in any random order, it seems that the attacker can improve guessing efficiency by exploring the partitions in decreasing order of probability. However, we observe a serious flaw in exploring the partitions in probability order. To see this, let $|\sigma_i|$ denote the size of the $i^{th}$ partition. Thus,
\begin{equation}
|\sigma| = \sum\limits_{i=1}^{n}|\sigma_i|
\end{equation}
In Figure~\ref{fig:partition}, the size of the partitions in the row number 4 is greater than the size of the partitions in every other row. Thus, if the utilized partitions are explored in decreasing order of probability, the attacker can first explore either the partition $(1,3)$ or the partition $(4,2)$ as both contain the same number of passwords. However, the attack is more efficient if the partition $(1,3)$ is explored first, since its size is smaller than the size of the partition $(4,2)$ and hence requires less computation. We capture this notion by defining the density of a partition.
\begin{definition}
The density $d_{\sigma_i}$  of a partition $\sigma_{i}$ is defined as the ratio of the number of passwords in a password database that belong to the partition $\sigma_{i}$ and the size of the partition $\sigma_{i}$.
\begin{equation}
d_{\sigma_i} =  \phi_i/|\sigma_i|
\end{equation}
\end{definition}
Since $\phi_{1,3} = \phi_{4,2}$ and $|\sigma_{1,3}| < |\sigma_{4,2}|$, the density of the partition $(1,3)$ is greater than the density of the partition $(4,2)$ $i.e$ $8/|\sigma_{1,3}| > 8/|\sigma_{4,2}|$ and therefore, the partition $(1,3)$ should be explored first. 

Thus, in the proposed attack model, the attacker divides the password search space into partitions and learns partition level features. The size of the resulting partitions depends upon the information available to the attacker. The attacker can now spend its resources more wisely by targeting just the utilized partitions. In this paper, we provide an optimal search strategy for a partition attacker. We demonstrate the effectiveness of a partition attacker and also show how these attacks can be countered. Specifically, our contributions are as follows.

\subsection{Contribution} 
\begin{enumerate}
\item We propose an offline attack model in which the attacker uses information available from previous breaches, attacks and surveys to learn partitions and to estimate partition level features such as partition densities and probabilities. The granularity of the resulting partitions depends upon the information available to the attacker. More information regarding the password database results in more refined partitions. We refer to this attack model as {\em space partition attack model} and the attacker as {\em partition attacker}. We prove that the success rate of a partition attacker is maximum if the resulting partitions are explored greedily in decreasing order of density. 

\item We show that the proposed attack model is more general as the existing well-known attack techniques including brute-force attacks, dictionary attacks~\cite{Klein}, grammar-based attacks~\cite{Weir}, probabilistic attacks~\cite{BonneauMetric} and attacks on random password generators~\cite{Ganesan} can be explained using our partition attacker framework. 

\item We demonstrate the effectiveness of the partition attack model with another instance of a partition attacker which we refer to as {\em bin attacker}. We measure the resistance of the publicly available breached databases and show that with the current computing power bin attacker can break nearly 90\% passwords in all password databases. Further, we hypothesize that the presence of composition policies still leads to non-uniform bin densities and hence can be exploited by bin attacker.

\item We prove that partition attackers can be countered only if the partition densities are uniform. 
We use this result and propose a {\em bin explorer system} to resists offline attacks by distributing users across different bins uniformly. 

\item We propose various bin assignment strategies for achieving uniform bin densities. We also derive password policy parameters such as password length based upon the computing capability of the partition attacker and the number of users in the system.

\item Finally, we adapt some of the well-known schemes to help users in choosing passwords from the system assigned bins and investigate their usability.
\end{enumerate}

\section{Space Partition Attack Model}
In the space partition attack model, the partition attacker uses information from previous leaks, surveys and attacks to divide the search space into partitions and to learn the partition level features such as partition densities and partition probabilities. The sizes of the resulting partitions need not be identical and can vary anywhere between 1 and the size of the search space $\sigma$. The granularity of partitions depends upon the information possessed by the attacker. The more information the attacker possesses, the more granular the partitions are. Now, these resulting partitions can be explored in various ways, the attacker can simply ignore the empty partitions and explore the non-empty partitions either in random order or in decreasing order of partition probabilities. However, we prove that the expected success of a partition attacker is maximum if the resulting partitions are explored in decreasing order of density.

Suppose that the attacker possesses the repository of passwords collected from various sources. Assume that these passwords are drawn from the password search space $\sigma$ and the attacker divides the search space $\sigma$ into $n$ partitions $\{\sigma_{1},\sigma_{2},...,\sigma_{n}\}$. Thus,
\begin{equation}
\sigma = \cup_{i=1}^{n} \sigma_{i}
\end{equation}  
Further, suppose this is the best refinement of the search space $\sigma$ which the attacker can achieve. We emphasize that the attacker has no further information regarding passwords inside the partitions. If the attacker had more information then these partitions could be refined further, thus contradicting the best refinement assumption. Therefore, the attacker assumes every password within a given partition to be equally likely.

 Let $\phi \gg n$ denote the total number of passwords in the password repository of attacker and $\phi_{i}$ denote the number of passwords that belong to the partition $\sigma_{i}$. Since the same password can be used by multiple users, $\phi_i$ represents the cardinality of a multiset, where each element of this multiset is a password drawn from the partition $\sigma_{i}$. 
\begin{equation}
\phi =  \sum\limits_{i=1}^n \phi_{i}
\end{equation}
Thus, the partition densities are $d_{\sigma} = \{\frac{\phi_{1}}{|\sigma_{1}|},\frac{\phi_{2}}{|\sigma_{2}|},\ldots,\frac{\phi_{n}}{|\sigma_{n}|}\}$ and the partition probabilities are $p_{\sigma} = \{\frac{\phi_{1}}{\phi},\frac{\phi_{2}}{\phi},\ldots,\frac{\phi_{n}}{\phi}\}$. 

Now, let $\alpha < |\sigma|$ represent the number of guesses that the attacker can compute to mount the offline attack. Since $\alpha < |\sigma|$ the attacker cannot mount the brute-force attack on the search space $\sigma$. The objective of the attacker is to maximize the success rate of the attack using the limited computational resource $\alpha$. Assume that the attacker distributes this resource $\alpha$ across $n$ partitions $\{\sigma_{1},\sigma_{2},...,\sigma_{n}\}$ as $\{\alpha_{1},\alpha_{2},...,\alpha_{n}\}$, {\em i.e.} the attacker spends an effort of $\alpha_i \le |\sigma_i|$ on the partition  $\sigma_i$. Thus,
\begin{equation}
\alpha = \sum\limits_{i=1}^{n} \alpha_i
\end{equation}
Since the attacker has only the partition level view and assumes every password within a given partition to be equally likely, the attacker can expect to recover $ \alpha_{i} \cdot \frac{\phi_{i}}{|\sigma_{i}|}$ passwords from the partition $\sigma_i$ after generating $\alpha_i$ guesses. Therefore, the expected success of a partition attacker is,
\begin{equation}\label{eq:success}
E(success) =  \sum\limits_{i=1}^{n} \alpha_{i} \cdot \frac{\phi_{i}}{|\sigma_{i}|}
\end{equation}
\subsection{Partition Attacker Characteristics}
In short, the partition attacker can perform the following tasks:
\begin{enumerate}
\item divide the password search space $\sigma$ into $n$ partitions $\{\sigma_{1},\sigma_{2},...,\sigma_{n}\}$.
\item compute partition level features such as partition densities \\$d_{\sigma} = \{\frac{\phi_{1}}{|\sigma_{1}|},\frac{\phi_{2}}{|\sigma_{2}|},\ldots,\frac{\phi_{n}}{|\sigma_{n}|}\}$ and partition probabilities $p_{\sigma} = \{\frac{\phi_{1}}{\phi},\frac{\phi_{2}}{\phi},\ldots,\frac{\phi_{n}}{\phi}\}$.
\item assume every password within a given partition $\sigma_{i}$  to be equally likely with a probability $\frac{\phi_{i}}{|\sigma_{i}|}$ .
\item decide upon a specific attack strategy and distribute the available computing power $\alpha < |\sigma|$ among $n$ partitions $\{\sigma_{1},\sigma_{2},...,\sigma_{n}\}$ as $\{\alpha_{1},\alpha_{2},...,\alpha_{n}\}$ .
\item explore the partition $\sigma_{i}$ by generating $\alpha_i$ random guesses without repetition.
\end{enumerate}

In the next section, we show how this partition attacker can spend the available resource $\alpha$ pragmatically and achieve the maximum expected success.
\subsection{Maximum Expected Success}
\newtheorem{Theorem}{Theorem}
\begin{Theorem}
Suppose that the partition densities are non-uniform. Further, without loss of generality assume that $\frac{\phi_{1}}{|\sigma_{1}|} \geq \frac{\phi_{2}}{|\sigma_{2}|} \geq \ldots \geq \frac{\phi_{n}}{|\sigma_{n}|}$. If the computational resource $\alpha$ available for an attack is limited, {\em i.e.} $\alpha < |\sigma|$, then the maximum expected success of a partition attacker is,
\begin{equation}
E_{max}(success) =  \sum\limits_{i=1}^{i_{0}} \phi_{i}  + (\alpha -  \sum\limits_{i=1}^{i_{0}} |\sigma_{i}|) \cdot \frac {\phi_{i_0+1}}{|\sigma_{i_0+1}|}
\end{equation}
where $i_{0}$ is such that $ \sum\limits_{i=1}^{i_{0}}{|\sigma_{i}|} \le \alpha < \sum\limits_{i=1}^{i_{0}+1}{|\sigma_{i}|}$.
\end{Theorem}
\begin{proof}
Let ,
\begin{equation}\label{eq:effort}
\alpha_{i} = \begin{cases} |\sigma_{i}| & 1 \le i \le i_{0} \\ 
\alpha -  \sum\limits_{i=1}^{i_{0}}{|\sigma_{i}|} & i = i_{0} + 1\\
0 & i_{0}+2 \le i \le n \end{cases} 
\end{equation}
Then, by equation (\ref{eq:success}) the expected success for the attacker is,
\begin{equation}
E(success) =  \sum\limits_{i=1}^{i_{0}} \phi_{i}  + (\alpha -  \sum\limits_{i=1}^{i_{0}} |\sigma_{i}|) \cdot \frac {\phi_{i_0+1}}{|\sigma_{i_0+1}|}
\end{equation}
Now, we show that the attacker cannot do better than this. Suppose that the same effort $\alpha$ is distributed across $n$ partitions of the search space $\sigma$ in other way, {\em i.e.} $\alpha = \sum\limits_{i=1}^{n} \alpha^{'}_{i}$. Then, by equation (\ref{eq:success}) we have,
\begin{equation}
E^{'}(success) =  \sum\limits_{i=1}^{n} \alpha^{'}_{i} \cdot \frac{\phi_{i}}{|\sigma_{i}|}
\end{equation}
Now we have,
\begin{equation}
\Delta = E(success)  - E^{'}(success)
\end{equation}
\begin{equation}
\Delta =  \sum\limits_{i=1}^{n}  (\alpha_{i} - \alpha^{'}_{i}) \cdot \frac{\phi_{i}}{|\sigma_{i}|}
\end{equation}
\begin{equation}\label{eq:ineq}
\begin{split}
\Delta =  \sum\limits_{i=1}^{i_{0}}  (\alpha_{i} - \alpha^{'}_{i}) \cdot \frac{\phi_{i}}{|\sigma_{i}|} +  (\alpha_{i_{0}+1} - \alpha^{'}_{i_{0}+1}) \cdot \frac{\phi_{i_{0}+1}}{|\sigma_{i_{0}+1}|} \\-  \sum\limits_{i=i_{0}+2}^{n}  (\alpha^{'}_{i} - \alpha_{i}) \cdot \frac{\phi_{i}}{|\sigma_{i}|}
\end{split} 
\end{equation}
Using equation (\ref{eq:effort}), we get $(\alpha_{i} - \alpha^{'}_{i})$ = $(|\sigma_{i}| - \alpha^{'}_{i}) \ge 0$, for $i \le i_{0}$. Also we know that $\frac{\phi_{i}}{|\sigma_{i}|} \ge \frac{\phi_{i_{0}+1}}{|\sigma_{i_{0}+1}|}$, for $i \le i_{0}$. 
\begin{equation}\label{eq:ref1}
\sum\limits_{i=1}^{i_{0}}  (\alpha_{i} - \alpha^{'}_{i}) \cdot \frac{\phi_{i}}{|\sigma_{i}|} \ge \sum\limits_{i=1}^{i_{0}}  (\alpha_{i} - \alpha^{'}_{i}) \cdot \frac{\phi_{i_0+1}}{|\sigma_{i_0+1}|}
\end{equation}
Again using equation (\ref{eq:effort}), we have $(\alpha^{'}_{i} - \alpha_{i})$ = $\alpha^{'}_{i} \ge 0$ for $i \ge i_{0}+2$, and we know that $\frac{\phi_{i_0+1}}{|\sigma_{i_0+1}|} \ge \frac{\phi_{i}}{|\sigma_{i}|}$ for $i \ge i_{0+2}$. 
\begin{equation}\label{eq:ref2}
\sum\limits_{i=i_{0}+2}^{n}  (\alpha^{'}_{i} - \alpha_{i}) \cdot \frac{\phi_{i_0+1}}{|\sigma_{i_0+1}|} \ge  \sum\limits_{i=i_{0}+2}^{n}  (\alpha^{'}_{i} - \alpha_{i}) \cdot \frac{\phi_{i}}{|\sigma_{i}|}
\end{equation}
Using equations (\ref{eq:ineq}), (\ref{eq:ref1}) and (\ref{eq:ref2}), we get,
\begin{equation}
\begin{split}
\Delta  \ge  \sum\limits_{i=1}^{i_{0}}  (\alpha_{i} - \alpha^{'}_{i}) \cdot \frac{\phi_{i_0+1}}{|\sigma_{i_0+1}|} +  (\alpha_{i_{0}+1} - \alpha^{'}_{i_{0}+1}) \cdot \frac{\phi_{i_{0}+1}}{|\sigma_{i_{0}+1}|} \\ - \sum\limits_{i=i_{0}+2}^{n}  (\alpha^{'}_{i} - \alpha_{i}) \cdot \frac{\phi_{i_0+1}}{|\sigma_{i_0+1}|}
\end{split} 
\end{equation}
\begin{equation}
\Delta \ge (\sum\limits_{i=1}^{n}  \alpha_{i} - \sum\limits_{i=1}^{n} \alpha^{'}_{i}) \cdot  \frac{\phi_{i_0+1}}{|\sigma_{i_0+1}|} = 0
\end{equation}
\noindent
Hence, the distribution of effort $\alpha$ as given in equation (\ref{eq:effort}) is optimal, {\em i.e.} {\em the benefit is maximum if the partitions are explored greedily in decreasing order of density.}
\end{proof}
\newtheorem{Observation}{Observation}
\begin{Observation}
If the size of all $n$ partitions is equal, {\em i.e.} $ |\sigma_{1}| = |\sigma_{2}| = \ldots = |\sigma_{n}|$, then 
\begin{equation}
i_{0} \le \alpha \cdot \frac{1}{\frac{|\sigma|}{n}}
\end{equation}
\end{Observation}
\newtheorem{Lemma}{Lemma}
\begin{Lemma}
The expected success of an attacker with the limited computational power $\alpha < |\sigma|$ when the partition densities are uniform, {\em i.e.} $\frac{\phi_{1}}{|\sigma_{1}|} = \frac{\phi_{2}}{|\sigma_{2}|} = \ldots = \frac{\phi_{n}}{|\sigma_{n}|}$ is
\begin{equation}
E_{uniform}(success) = \frac{\phi}{|\sigma|} \cdot \alpha
\end{equation}
\end{Lemma}
\begin{proof}
By equation (\ref{eq:success}) we get,
\begin{equation}\label{eq:lem1}
E_{uniform}(success) =  \sum\limits_{i=1}^{n} \frac{\phi_{i}}{|\sigma_{i}|}\cdot \alpha_{i}
\end{equation}
As the partition densities are uniform we have,
\begin{equation}\label{eq:lem2}
\frac{\phi_{1}}{|\sigma_{1}|} = \frac{\phi_{2}}{|\sigma_{2}|} = \ldots = \frac{\phi_{n}}{|\sigma_{n}|} = c 
\end{equation}
Thus, $\phi_{1} = c \cdot |\sigma_{1}|$  , $\phi_{2} = c \cdot |\sigma_{2}|$, $\ldots$ , $\phi_{n} = c \cdot |\sigma_{n}|$. Adding these we get,
\begin{equation}\label{eq:lem3}
c = \frac{\sum\limits_{i=1}^n \phi_{i}}{\sum\limits_{i=1}^n |\sigma_{i}|} = \frac{\phi}{|\sigma|}
\end{equation}
Using equations (\ref{eq:lem1}), (\ref{eq:lem2}) and (\ref{eq:lem3}) we get,
\begin{align}
E_{uniform}(success) &=  \sum\limits_{i=1}^{n} \frac{\phi}{|\sigma|} \cdot \alpha_{i}\\
E_{uniform}(success) &=  \frac{\phi}{|\sigma|} \cdot \sum\limits_{i=1}^{n} \alpha_{i}\\
E_{uniform}(success) &=  \frac{\phi}{|\sigma|} \cdot \alpha
\end{align}
\end{proof}
\noindent
Thus, if the partition densities are uniform then the attacker with the limited computational resource $\alpha$ cannot expect to get more than $\frac{\phi}{|\sigma|} \cdot \alpha$ passwords.

Another partition level feature that partition attacker can use is the partition probability. The partition probabilities are given by $p_\sigma=\{\frac{\phi_{1}}{\phi},\frac{\phi_{2}}{\phi},\ldots,\frac{\phi_{n}}{\phi}\}$. The partition $\sigma_k$ is said to be more popular than another partition $\sigma_i$ if $\phi_k > \phi_i$.  Thus, exploring the partitions in decreasing order of probabilities also seems to be a good attack strategy. But, we prove that the maximum benefit is achieved only if the computational resource $\alpha$ is spent to explore the denser partitions rather than the popular partitions.

\begin{Lemma}
Suppose that the partition densities are non-uniform. Further, without loss of generality assume that $\frac{\phi_{1}}{|\sigma_{1}|} \geq \frac{\phi_{2}}{|\sigma_{2}|} \geq \ldots \geq \frac{\phi_{n}}{|\sigma_{n}|}$. If $\exists k > j$ such that $\sum\limits_{i=1}^{j} |\sigma_{i}| \ge |\sigma_{k}|$, then $\sum\limits_{i=1}^{j} \phi_{i} \ge \phi_{k}$.
\end{Lemma}
\begin{proof}
Given that,
\begin{equation}
\frac{\phi_{1}}{|\sigma_{1}|} \ge \frac{\phi_{2}}{|\sigma_{2}|} \ge \ldots \ge \frac{\phi_{j}}{|\sigma_{j}|} \ldots \ge \frac{\phi_{k}}{|\sigma_{k}|}
\end{equation}
Therefore, we have $\phi_{1}\cdot |\sigma_{k}| \ge \phi_{k}\cdot |\sigma_{1}|$ , $\phi_{2}\cdot |\sigma_{k}| \ge \phi_{k}\cdot |\sigma_{2}|$ , $\ldots$ , $\phi_{j}\cdot |\sigma_{k}| \ge \phi_{k}\cdot |\sigma_{j}|$. 
Adding these we get,
\begin{equation}
|\sigma_{k}| \cdot \sum\limits_{i=1}^{j} \phi_{i} \ge \phi_{k} \sum\limits_{i=1}^{j} |\sigma_{i}|
\end{equation}
\begin{equation}
\frac{\sum\limits_{i=1}^{j} \phi_{i}}{\sum\limits_{i=1}^{j} |\sigma_{i}|} \ge \frac{\phi_{k}}{|\sigma_{k}|} 
\end{equation}
Since $\sum\limits_{i=1}^{j} |\sigma_{i}| \ge |\sigma_{k}|$,  hence $\sum\limits_{i=1}^{j} \phi_{i} \ge \phi_{k}$
\end{proof}
\noindent
 Suppose that $\phi_k \ge \phi_i,\forall i\le j$, {\em i.e.} the $k^{th}$ partition is more popular than any of the first $j$ partitions. Further suppose that the attacker can explore either the first $j$ denser partitions or the $k^{th}$ denser partition (which is also the most popular partition), {\em i.e.}, $\alpha =  \sum\limits_{i=1}^{j} |\sigma_{i}| = |\sigma_{k}|$. Then $Lemma \ 2$ implies that exploring all $j$ partitions is more beneficial than exploring the $k^{th}$ partition since we are guaranteed to have $\sum\limits_{i=1}^{j} \phi_{i} \ge \phi_{k}$. In other words, {\em the benefit is maximum if the attacker utilizes the available computational resource $\alpha$ to explore the denser partitions rather than the popular partitions.} 
\\
\\
\textbf{Reduction to the Fractional Knapsack Problem.} The partition attacker's problem of distributing the limited computational resource $\alpha$ across $n$ partitions $\{\sigma_1, \sigma_2,\ldots,\sigma_n\}$ to maximize the password guessing success can be reduced to the fractional knapsack problem~\cite{Goodrich} as shown in Table~\ref{tab:knapsack}. Here, the computational resource $\alpha$ available to the attacker is analogous to the capacity $W$ of the knapsack. Further,  every partition $\sigma_i$ can be viewed as an object $o_i$, with the size $|\sigma_i|$ of the partition $\sigma_i$ representing the weight $w_i$ of the object $o_i$ and the number of passwords $\phi_i$ representing the value $v_i$ associated with the object $o_i$. In the partition attack model, the goal is to recover the maximum number of passwords $\sum\limits_{i=1}^{n} \alpha_{i} \cdot \frac{\phi_i}{|\sigma_i|}$, where $\alpha_i \le |\sigma_i|$,  subject to the computational constraint  $\alpha \ge \sum\limits_{i=1}^{n} \alpha_{i}$. Analogously in the knapsack problem, the goal is to maximize the total value $\sum\limits_{i=1}^{n} x_i \cdot \frac{v_i}{w_i}$, where $x_i \le w_i$, subject to the capacity constraint $W \ge \sum\limits_{i=1}^{n} x_{i}$. The fractional knapsack problem has a polynomial time solution which can be obtained by sorting objects in decreasing order of $\frac{v_i}{w_i}$ and then picking the objects in a sequence until the knapsack is full. Similarly, the optimal attack strategy of a partition attacker is to explore the partitions in decreasing order of density $\frac{\phi_i}{|\sigma_i|}$ till the computational resource  $\alpha$ is completely exhausted. Due to sorting, the algorithm takes time $O(n\cdot log n)$. However, by adapting the algorithm to find weighted medians~\cite{Korte}, this problem can be solved in linear time $O(n)$.

\begin{table}[h]
\scriptsize
\centering
\caption{Analogy between the Partition Attacker and the Fractional Knapsack Problem.}~\label{tab:knapsack}
\begin{tabular}{|c|c|c|} \hline
Category & Partition Attacker & Fractional Knapsack\\ \hline
Capacity	     & $\alpha$ & $W$\\
Objects               & $\{\sigma_1, \sigma_2,\ldots,\sigma_n\}$ & $\{o_1, o_2,\ldots,o_n\}$	\\
Weights                & $\{|\sigma_1|, |\sigma_2|,\ldots,|\sigma_n|\}$ & $\{w_1, w_2,\ldots,w_n\}$\\	
Values                & $\{\phi_1, \phi_2,\ldots,\phi_n\}$ & $\{v_1, v_2,\ldots,v_n\}$\\
Constraint               & $\alpha \ge \sum\limits_{i=1}^{n} \alpha_{i}$  & $W \ge \sum\limits_{i=1}^{n} x_{i}$\\		
Maximize & $\sum\limits_{i=1}^{n} \alpha_{i} \cdot \frac{\phi_i}{|\sigma_i|}$ & $\sum\limits_{i=1}^{n} x_i \cdot \frac{v_i}{w_i}$\\
\hline
\end{tabular}
\end{table}

\section{Generality of the Attack Model}
We emphasize that the space partition attack model is more general. The partition attacker uses information from various sources to refine partitions and estimate their densities. Subsequently, the attacker recovers passwords from the target database by exploring the partitions in decreasing order of density. We demonstrate the generality of the space partition attack model by explaining the existing, well-known attacks using our framework. Specifically, we show that brute-force attacks, dictionary-based attacks~\cite{John,Hashcat2}, grammar-based attacks~\cite{Weir}, probabilistic attacks~\cite{BonneauMetric} and attacks on random password generators~\cite{Ganesan} are different instances of the partition attack model.

\subsection{Brute-Force Attacks}
The brute-force attacker possesses no information regarding the target database and treats the entire search space as one big partition. In this case, the attacker assumes every password to be equally likely and enumerates all possible password combinations for guessing. This attack is guaranteed to recover every password from the system, however, it is generally computationally infeasible to generate and verify every possible combination against the target database. This single partitioning instance covers the brute-force attack model.

\subsection{Dictonary Attacks}
Users do not choose passwords uniformly~\cite{Morris,Klein,Moshe,Florencio}. They often choose simple passwords based on familiar words. The password guessing tools such as John the Ripper (JTR)~\cite{John} and Hashcat~\cite{Hashcat2} exploit this observation and use carefully crafted dictionaries and mangling rules (learned from previous breaches) to perform password guessing. These so called dictionary attacks proceed as follows. First, the entire dictionary of size $N$ is matched against the target password database. Then, the most popular mangling rule is applied to every dictionary word and subsequently the resulting mangled dictionary is verified against the target database. This process is repeated for every mangling rule. We note that there is an implicit assumption regarding the density of the partitions. The partition containing the dictionary words is assigned the highest density and is therefore guessed first. Then the most popular mangling rule is applied on the dictionary to generate the second dense partition and so on.

For instance, consider a dictionary $W = \{lion, deer, tiger, horse\}$ of size $N = 4$. Let $R = \{W, W1, W!, W12\}$ be the set of mangling rules with the probability distribution $P = \{\frac{\phi_1}{\phi}=0.4, \frac{\phi_2}{\phi}=0.3, \frac{\phi_3}{\phi}=0.2,\frac{\phi_4}{\phi}=0.1\}$. JTR performs guessing as follows. First, it verifies the entire dictionary $W$ against the target database, then it appends ‘1’ to the dictionary and verifies the mangled dictionary $W1 = \{lion1, deer1, tiger1, horse1\}$, then it guesses $W! = \{lion!, deer!, tiger!$ $,horse!\}$ and finally $W12 = \{lion12, deer12, tiger12, horse12\}$ is verified. As all partitions $W$, $W1$, $W!$ and $W12$ are of equal size $|\sigma_i|=4$, the order of guesses generated in the probability order $\frac{\phi_i}{\phi}$ is same as those generated in the density order $\frac{\phi_i}{\sigma_i}$.

\subsection{Grammar-based Attacks}
In order to further speed-up the password guessing process, Weir {\em et al.}~\cite{Weir} proposed an online algorithm to derive a probabilistic context free grammar from the training dataset. The algorithm generates most effective word-mangling rules and subsequently applies them on a dictionary in decreasing order of probability. The authors introduced the concepts of base structure, pre-terminal structure and terminal structure to capture the common patterns in the passwords. For instance, $S_1L_3D_2$ is the base structure, representing passwords that begins with 1 symbol followed by 3 lowercase letters and ending with 2 digits. Substituting the values of symbol and digits in a base structure results in a pre-terminal structure {\em e.g.,} $\$L_312$ and substituting the values of alpha variable in a pre-terminal structure with a dictionary word produces a terminal structure {\em e.g.,} $\$cat12$ (Table~\ref{tab:grammar}). Terminal-structures are the actual password guesses.
\begin{table}[h]
\scriptsize
\centering
\caption{An example of password grammar.}~\label{tab:grammar}
\begin{tabular}{|l|c|c|} \hline
Structure	&	Example \\ \hline
Base	&        $S_1L_3D_2$ \\
Pre-Terminal	&	$\$L_312$ \\
\textbf{Terminal (Guess)}	& $\mathbf{\$cat12}$ \\
\hline
\end{tabular}
\end{table}

Weir{\em et al.}~\cite{Weir} proposed and compared the effectiveness of two attack strategies with the popular password guessing tool John the Ripper (JTR). The first attack strategy generates guesses in {\em pre-terminal probability order} and the second strategy generates guesses in {\em terminal probability order}. They simulated attacks on 3 different real-world password lists and found that the strategy based on {\em terminal probability order} was more efficient than JTR while the strategy based on {\em pre-terminal probability order} performed similar (in some cases worse) to JTR. They provided no explanation for this behavior, rather they were surprised by this result. They stated the following  {\em ``A surprising result to us was that when we used pre-terminal probability order, it did not result in a noticeable improvement over John the Ripper’s default rule set."}

We use our partition-based attacker framework to explain this behavior. {\em We show that the strategy based on pre-terminal probability order generates pre-terminal guesses in decreasing order of probability, while the second strategy based on terminal probability order actually generates pre-terminal guesses in decreasing order of density.} We have already proved that generating guesses in a density order is the most efficient strategy ($Theorem \ 1$). Thus, we just need to show that the {\em terminal probability-based attack}  generates guesses in same order as that of {\em pre-terminal density-based attack}.
\\
\textbf{Pre-terminal probability-based attack.} The first attack strategy as proposed by Weir {\em et al.}~\cite{Weir} is based on the probabilities of pre-terminal structures in which the attacker enumerates pre-terminals in decreasing order of probability. We call this strategy as {\em pre-terminal probability-based attack}. The attack proceeds as follows.
\begin{enumerate}
\item First, the attacker learns all pre-terminal structures and their probabilities from the training set of previously disclosed passwords. Here pre-terminals are the partitions.
\item Subsequently, the attacker finds a suitable dictionary for substituting alpha variables in the pre-terminal structures.
\item Finally, the attacker explores all pre-terminal structures in decreasing order of probability. More precisely, the attacker substitutes alpha variables in the most probable pre-terminal structure with the appropriate length dictionary words, verifies all resulting terminal structures (password guesses) against the target database, removes this pre-terminal structure from the queue and then repeat these guess-verify-remove process with the remaining pre-terminal structures.
\end{enumerate}

For illustration purpose, consider a toy example as shown in Figure~\ref{fig:preterminal}. The attacker uses previously breached databases and learns the probabilities of three pre-terminals $L_6\$1$, $L_5!$ and $\$L_412$ to be 5/10, 3/10 and 2/10 respectively. Subsequently, the attacker chooses dictionary words of length 6, 5 and 4 to replace the alpha variables $L_6$, $L_5$ and $L_4$ as shown in the figure. According to the {\em pre-terminal probability-based attack strategy}, the attacker simply explores all pre-terminals in decreasing order of probability.  The attacker generates guesses by substituting the alpha variable $L_6$ in the most probable pre-terminal $L_6\$1$ with 8 dictionary words, then by substituting the alpha variable $L_5$ in the pre-terminal $L_5!$ with 4 dictionary words and finally by substituting the alpha variable $L_4$ in the least probable pre-terminal $L_412$ by 2 dictionary words. Thus, as per pre-terminal probability-based attack guesses are generated in the following order,
\begin{enumerate}
\item $L_6\$1$ :   $\{monkey\$1, donkey\$1, jaguar\$1, rabbit\$1, turtle\$1, python\$1,$\\ $falcon\$1, parrot\$1\}$
\item	$L_5!$ :     $\{tiger!, horse!, zebra!, sheep!\}$
\item	$L_412$ : $\{\$lion12, \$deer12\}$
\end{enumerate}
\begin{figure}[h]
\centering
\includegraphics[scale=0.5]{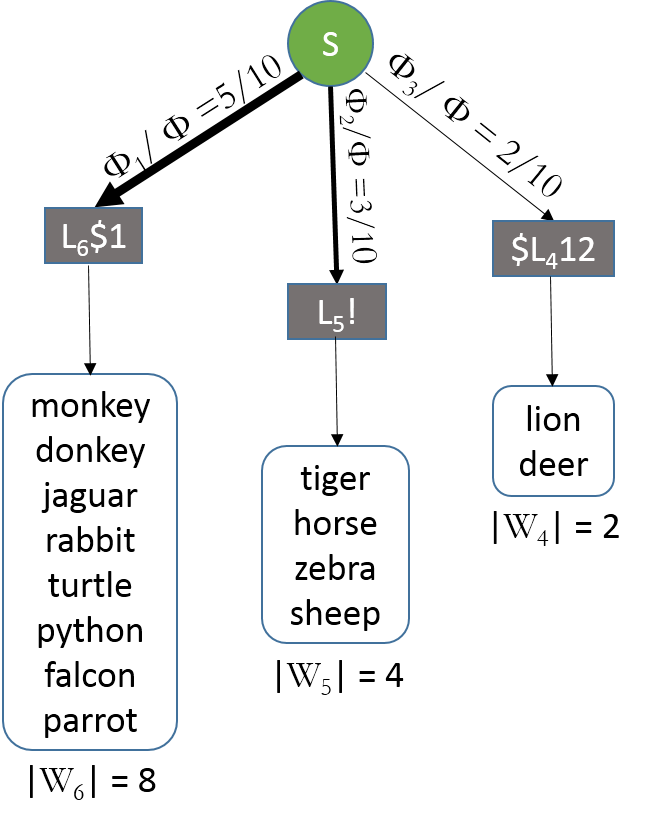}
\caption{The pre-terminal pobability-based attack as proposed by Weir {\em et al.}~\cite{Weir}. The pre-terminal $L_6\$1$ has the highest probability followed by $L_5!$ and $\$L_412$.}~\label{fig:preterminal}
\end{figure}
\noindent
\textbf{Terminal Probability Order.} The second attack strategy as proposed by Weir {\em et al.}~\cite{Weir} is based on the probabilities of terminal structures, in which the attacker generates guesses in decreasing order of terminal probability. In this strategy, the attacker also assigns the probability to every dictionary word. The attacks in~\cite{Weir} assumed all dictionary words of the same length to be equally likely. Thus, in our toy example, the probability of a 6 length dictionary word is 1/8 as there are 8 such words, similarly the probability of a 5 length dictionary word is 1/4 and 4 length dictionary word is 1/2. The terminal structure is obtained by replacing alpha variables in pre-terminal structures. Thus, the probability of a terminal is the product of probability of its pre-terminal structure and the probability assigned to a dictionary word used in the replacement of alpha variable. If all dictionary words having the same length are equally likely then all terminal structures (password guesses) generated using a given pre-terminal structure have the same probability.

To see this, let $\phi$ be the number of passwords in the training set and let $\phi_1/\phi, \phi_2/\phi, \ldots, \phi_n/\phi$ be the probabilities of the $n$ pre-terminal structures. Further suppose that $|W_m|$ is the total number of $m$ length words in the dictionary. The probability of $m$ length dictionary word as assigned by Weir {\em et al.}~\cite{Weir} is thus $1/|W_m|$. As the probability of all $m$ length dictionary words is uniform $1/|W_m|$ and the probability of the $i^{th}$ pre-terminal structure is  $\phi_i/\phi$, all terminal structures (password guesses) generated by substituting the alpha variable with a dictionary word in the $i^{th}$ pre-terminal structure have the same probability $(\phi_i/\phi)\cdot (1/|W_m|)$. The factor $1/\phi$ is associated with every terminal structure, hence we can simply drop this term, yielding the value $\phi_i/|W_m|$, which is nothing but the density of the $i^{th}$ pre-terminal. In other words, the terminal probability-based strategy basically generates password guesses in decreasing order of pre-terminal density (Figure~\ref{fig:terminal}).
{\em Thus, the terminal probability-based attack is essentially the pre-terminal density-based attack which Weir et al.~\cite{Weir} found to be more efficient than the pre-terminal probability-based attack}.

In our toy example, the order of guesses generated according to the pre-terminal density-based attack is as follows,

\begin{enumerate}
\item $\$L_412$ - $\{\$lion12, \$deer12\}$
\item 	$L_5!$ -      $\{tiger!, horse!, zebra!, sheep!\}$
\item $L_6\$1$  - $\{monkey\$1, donkey\$1, jaguar\$1, rabbit\$1, turtle\$1, python\$1,$\\$falcon\$1, parrot\$1\}$
\end{enumerate}

\begin{table}[h]
\scriptsize
\centering
\caption{Probability and density of pre-terminals in the toy example.}~\label{tab:preterminal}
\begin{tabular}{|c|c|l|} \hline
Pre-terminal	& Probability $\frac{\phi_i}{\phi}$ & Density $\frac{\phi_i}{|W_m|}$\\ \hline
$L_6\$1$	&  \textbf{5/10} & 5/8\\
$L_5!$	&  3/10 & 6/8 = 3/4 \\
$\$L_412$	&  2/10 & \textbf{8/8 = 2/2 = 1} \\
\hline
\end{tabular}
\end{table}
The value of parameters used in the toy example is given in Table~\ref{tab:preterminal}. The guesses generated using the pre-terminal density-based attack are in the exact reverse order of those generated using the pre-terminal probability-based attack.
\begin{figure}[h]
\centering
\includegraphics[scale=0.6]{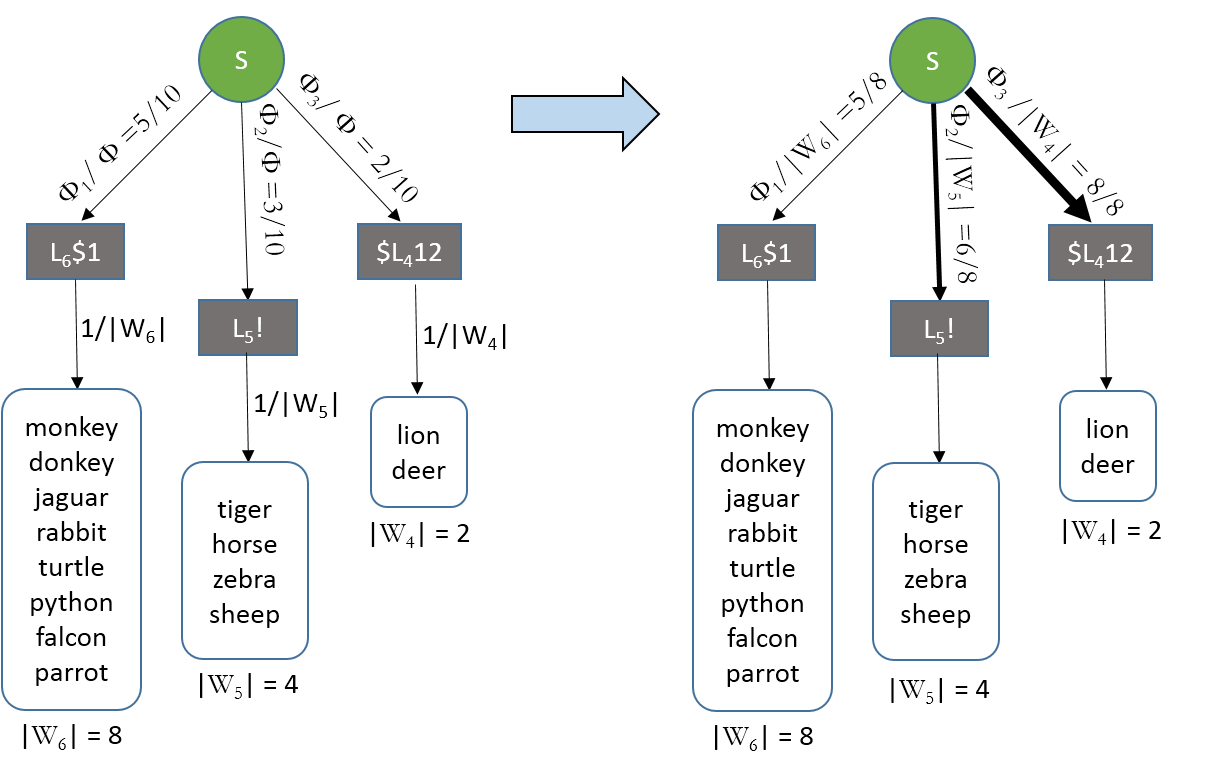}
\caption{The terminal probability-based attack as proposed by Weir {\em et al.}~\cite{Weir} is equivalent to the pre-terminal density-based attack. The pre-terminal $\$L_412$ has the highest density followed by $L_5!$ and $\L_6\$1$.}~\label{fig:terminal}
\end{figure}

\subsection{Probabilistic Attacker}
Consider a powerful-attacker (omniscient attacker) who has complete information regarding the target password database and divides the search space into unit-sized partitions (every partition has only 1 password). By $Theorem\  1$, the maximum benefit is achieved if these unit-sized partitions  are explored in decreasing order of density. In this particular case, it does not matter if the guesses are generated using partition densities $d_{\sigma} = \{\frac{\phi_{1}}{1},\frac{\phi_{2}}{1},\ldots,\frac{\phi_{n}}{1}\}$ or partition probabilities $p_{\sigma} = \{\frac{\phi_{1}}{\phi},\frac{\phi_{2}}{\phi},\ldots,\frac{\phi_{n}}{\phi}\}$, as both these strategies produce guesses in the same order. This specific instance covers Bonneau's attack model~\cite{BonneauMetric}. 

\subsection{Sandia Attack}
The offline attacks mainly exploit the fact that the human-generated passwords have non-uniform distribution. However, due to poor implementation even the random password generators can be flawed. In 1993, Ganesan and Davies~\cite{Ganesan} showed that two pronounceable password generators, namely Sandia system and NIST system, are not as secure as they claim. The Sandia system used 25 different pronounceable templates for generating a pronounceable password, {\em e.g., cvcvcvc} is one such template, where {\em c} stands for any consonant and {\em v} for any vowel. These 25 templates were of different sizes, but for the password generation they were chosen with equal probability. Consequently, all 25 templates had roughly the same number of passwords which lead to non-uniform template densities. Thus, in the case of Sandia system every template represents one partition and the attacker~\cite{Ganesan} can exploit the resulting non-uniform partition densities. Similar flaw was observed in the NIST system as well.

\subsection{Entropy-based Password Checker}
In 2001, Yan~\cite{Yan:2001:NPP:508171.508194} proposed an entropy-based password checker to prevent passwords with low entropy. The main idea of this checker was to divide password space into partitions (patterns) and then to classify passwords drawn from smaller partitions as weak. The authors illustrated the working of their checker by concentrating on 7 character alpha-numeric passwords. The total search space for a 7 character password composed using lowercase letters and digits is $36^7$. The partitions were formed based on the number of distinct lowercase letters in the password. The passwords composed of either digits, {\em e.g.,} ``1234567" or repeating a lowercase letter multiple times, {\em e.g.,} ``abbbb12" can be searched quickly ($\sim 10^7$ guesses). Whereas the passwords composed of either 7 distinct lowercase letters,  {\em e.g.,} ``jdowsna" or 6 distinct lowercase letters and 1 digit,  {\em e.g.,} ``is8kdab" requires more effort ($\sim {26 \choose 7} \cdot 7!$ guesses). Thus, passwords created using distinct letters belong to large partitions and are classified as strong, whereas passwords containing either repetitive letters or digits belong to smaller partitions and are classified as weak. 

\begin{table}[h]
\scriptsize
\centering
\caption{Different instances of a partition attacker. The partitions created by probabilistic attacker are the most granular (unit-sized) while those created by brute-force attacker are least granular (only one partition).}~\label{tab:attacker}
\begin{minipage}{8cm}
\begin{tabular}{|c|c|c|} \hline
Partition Attacker &  Partition Example\\ \hline
Brute-force & entire search space\\
Yan pattern &   passwords containing 5 lowercase letters and 2 digits\\
Bin \footnote{The bin attacker is explained in the next section.} &   $L_5D_2$\\
Pre-terminal &  $L_512$\\
Hybrid Bin & popular words such as $abcde12$ followed by bin $L_5D_2$\\
Probabilistic  & $abcde12$\\
\hline\end{tabular}
\end{minipage}
\end{table}

In the following section, we present another instance of a partition attacker, {\em bin attacker}, and show that it is much effective than the brute-force attacker. Specifically, we divide the search space  into {\em password bins} and demonstrate that the bin densities in real-world password databases are non-uniform which results in a huge benefit for the attacker.
\section{Bin Attacker}~\label{sec:binattack}
Before defining the bin attacker, we introduce the notation which will be used in the rest of the paper.
\subsection{Notation}
\noindent
$L$ - An alphabet representing the set $\{a,\ldots,z\}$ of lowercase letters.\\ 
$U$ - An alphabet representing the set $\{A,\ldots,Z\}$ of uppercase letters. \\
$D$ - An alphabet representing the set $\{0,\ldots,9\}$ of decimal digits. \\
$S$ - An alphabet representing the set of 33 special symbols such as @,\#,\$,\&,* and so on. \\
$+$ - denotes 1 or more occurrences of the alphabet. \\
$*$ -  denotes 0 or more occurrences of the alphabet. \\
$?$ -  denotes 0 or 1 occurrence of the alphabet. \\
$\{i,j\}$ - denotes at least $i$ and at most $j$ occurrences of the alphabet, where $0 \le i \le j$.\\
{$|\lambda|$} -  represents the number of elements in a set $\lambda$.
\indent
\subsection{Bins}
\begin{figure}[h]
\centering
\includegraphics[scale=0.5]{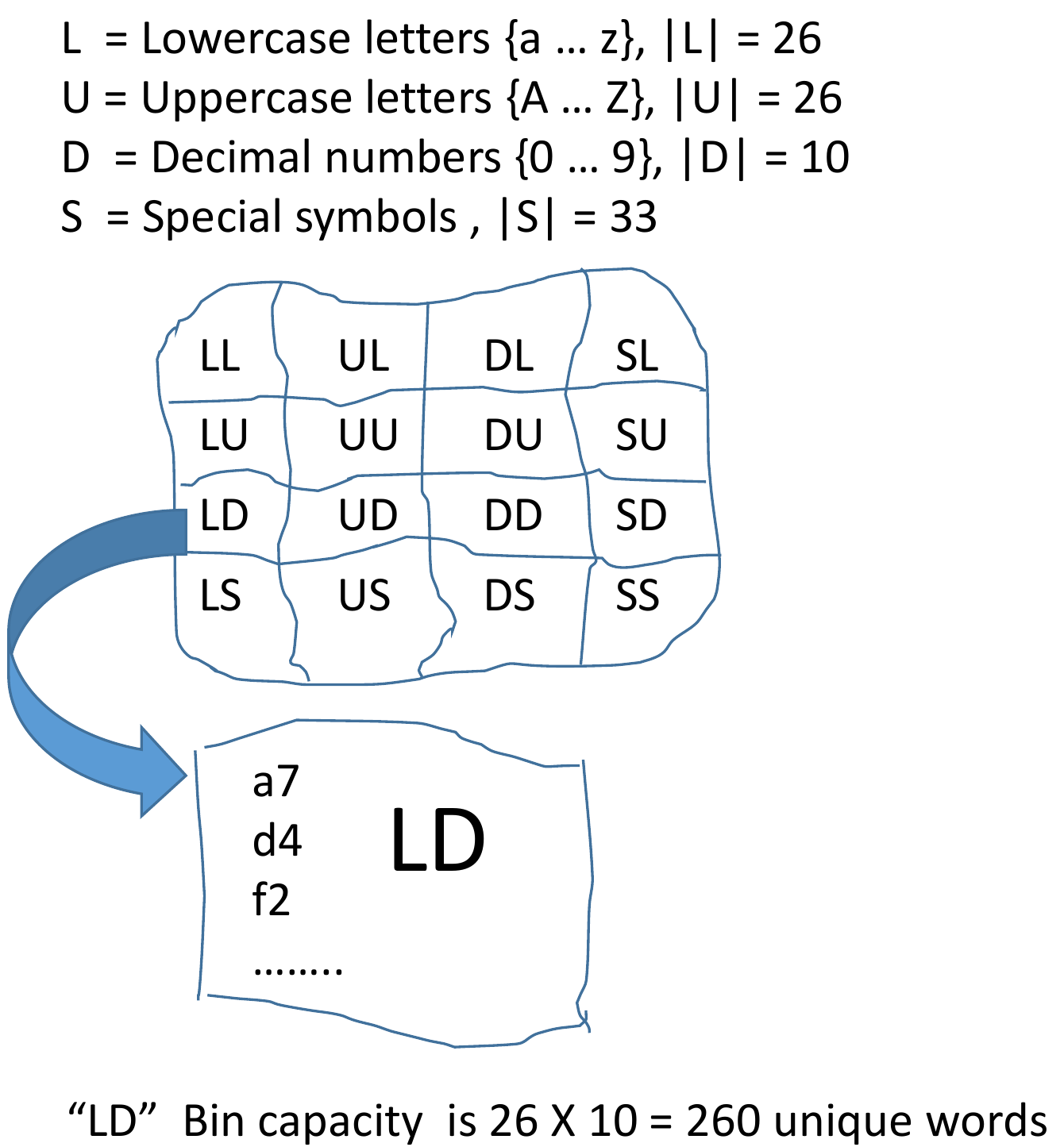}
\caption{Illustration of a bin space. The number of $l=2$ length bins derived using the alphabet set $\sigma = \{L,U,D,S\}$  is $4^l=4^2=16$.}~\label{fig:bin}
\end{figure}
Consider an alphabet set $\gamma= \{L,U,D,S\}$ composed of 4 character classes. We refer to the $l$ length strings derived from $\gamma$ as {\em password bins} or {\em bins}, $e.g.,$ $L_8$, $U_1L_7$, $S_1U_1L_5D_1$ are 8 length bins.  Every bin represents a class of passwords, {\em e.g.,} the bin $L_8$ represents 8 length passwords composed entirely of lowercase letters while the bin $U_1L_7$ represents 8 length passwords that begin with an uppercase letter followed by 7 lowercase letters. 
Thus, the password search space is divided into total $4^{l}$ bins. We define the {\em bin capacity} as the number of passwords represented by the bin (Figure~\ref{fig:bin}). For instance, the capacity of the bin $L_8$, composed of all lowercase letters is $|L_8|=26^8$. By summing the capacities of all bins we obtain,
\begin{align}
\nonumber  \sum\limits_{k=1}^{4^l}C_k &= (|L|+|U|+|D|+|S|)^l\\
					&= |\gamma|^l
\end{align}
where $C_k$ is the capacity of $k^{th}$ bin. Note that the number of bins increases exponentially with the increase in password length $l$. 
Let the maximum password length be $l_{max}$. Therefore, the size of the search space $\sigma$ is,
\begin{align}\label{eq:ss}
\nonumber |\sigma| &= \sum\limits_{l=1}^{l_{max}} |\gamma|^{l}\\
\nonumber             & = \frac{|\gamma|^{l_{max}+1} - |\gamma|}{|\gamma|-1}\\
 |\sigma|	&  \approx |\gamma|^{l_{max}}
\end{align}
\begin{definition}
The {\em bin attacker} is an instance of a partition attacker which (1) divides the password search space $\sigma$ into bins and (2) explores these bins in decreasing order of density (by $Theorem\ 1$).
\end{definition}
We show that the bin attacker performs much better than the brute-force attacker by exploring only the denser bins ($Theorem \ 1$) and completely ignoring the unutilized bins. In the next sections, we discuss the training and test datasets and then demonstrate the effectiveness of the bin attacker.
\subsection{Training and Test Data}
\begin{table}[h]
\scriptsize
\centering
\caption{Total number of passwords and utilized bins in publicly available breached databases.}~\label{tab:dataset}
\begin{tabular}{|c|c|c|c|c|c|} \hline
No & Category & Database & Total Passwords & Utilized Bins & Year\\ \hline
1 & Gaming & Rockyou	& 32,603,388 & 140,401 & 2009\\
2 & Gaming & Gamigo     & 6,919,630  & 40,365    & 2012\\
3 & Social & LinkedIn       & 5,586,887  & 102,066  & 2012\\
4 & Mail & Gmail		& 4,929,090  & 23,348    & 2014\\
5 & Mail  & Mail.ru 	& 4,664,479  & 83,061    & 2014\\
6 & Mail & Yandex 	& 1,261,810  & 26,428    & 2014\\
7 & Media & Gawker	& 1,085,085  & 5,539      & 2010 \\
8 & Misc & Yahoo		& 442,834     & 15,711    & 2012\\
9 & Software & Phpbb	& 255,421     & 7,905      & 2009\\
\hline\end{tabular}
\end{table}
\noindent
We collected 9 publicly available breached password databases~\cite{Adepts,Skull}. We were particularly interested in studying the passwords of compromised websites with a large user base. Table~\ref{tab:dataset} shows the number of passwords along with the category of breached websites and the year of breach. The table also depicts the number of bins that were used in each of the databases for creating passwords, $e.g.,$ $32.6 \ million$ passwords in the Rockyou database were derived from just 140,401 bins. These bins were not used uniformly and as depicted in Figure~\ref{fig:lendist} more than 80\% of the passwords in all breached databases were shorter ($l \le 10$). Further, none of these websites had any password composition requirements during the time of breach. The breach of these databases also revealed the insecure storage practices followed by the websites. For instance, all Rockyou passwords were stored in plaintext while longer Gawker passwords were truncated to 8 characters before hashing. {\em None of the password databases studied in this paper had their passwords salted. They were either stored in plaintext or were at most hashed.} 

In the following section, we demonstrate how the bin attacker can guess a large proportion of passwords in each of the breached databases. We use the Rockyou database to estimate the bin densities and then measure the resistance of the remaining 8 password databases against the bin attacker.
\begin{figure}[h]
\centering
\includegraphics[scale=0.27]{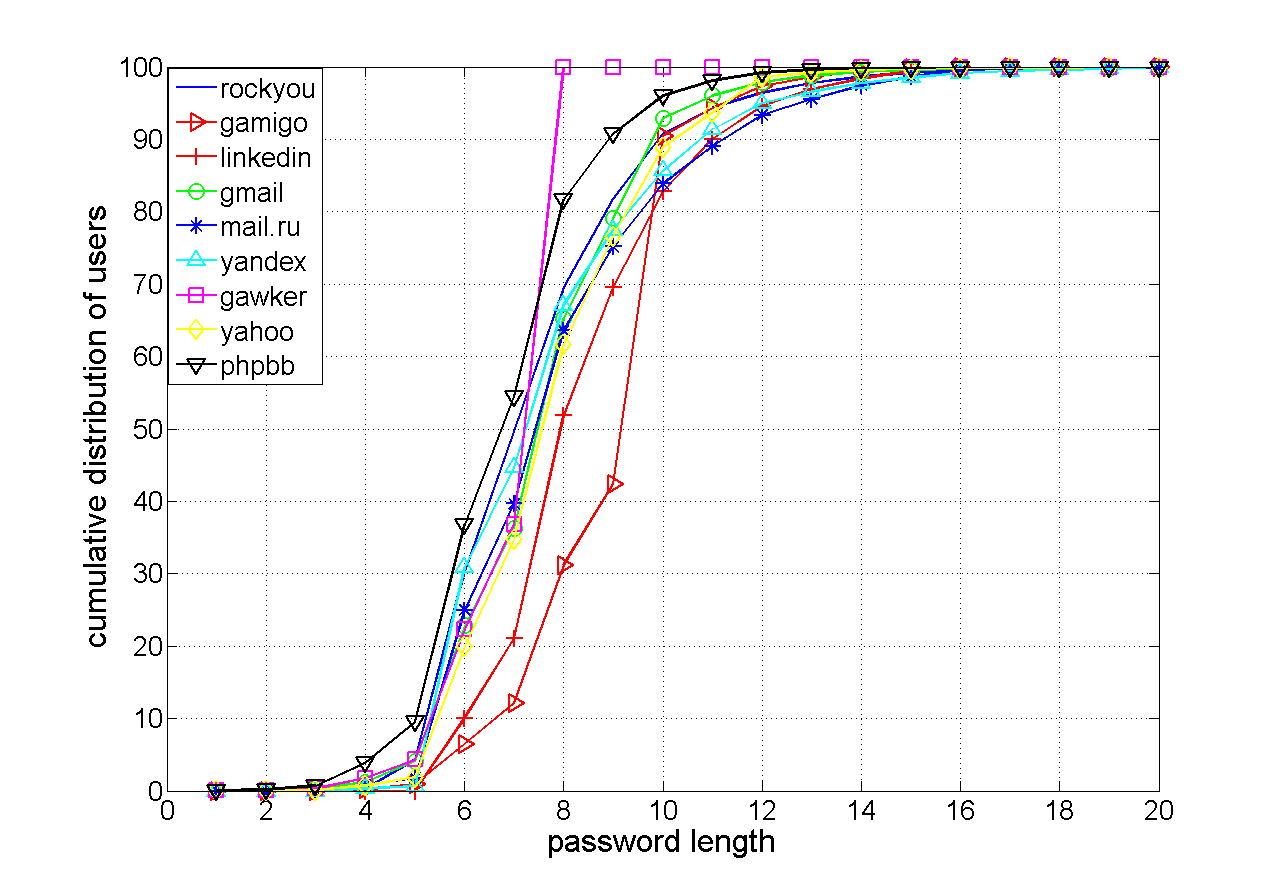}
\caption{Length-wise cumulative distribution of passwords in breached databases. More than 90\% passwords have length at most 12 ($l \le 12$).}~\label{fig:lendist}
\end{figure}
\subsection{Efficiency of the Bin Attacker}
Consider a system with $32 \ million$ registered users, $\phi = 2^{25}$. Suppose that the maximum password length $l_{max}$ is 10 and thus by equation (\ref{eq:ss}) the attack search space is $|\sigma| \approx 95^{10} = 2^{65.7}$. Today, a dedicated password cracking hardware can generate up to 350 billion ($2^{38.35}$) guesses per second~\cite{GPU}, we assume that the bin attacker can use such machine for 2.5 days ($2^{17.72}$ seconds) and generate $\alpha = 2^{38.35} \cdot 2^{17.72} \approx 2^{56}$ guesses. If the bin densities are uniform then by $Lemma \ 1$ the expected success of this bin attacker is,
\begin{equation}\label{eq:32m}
E_{uniform}(success) = \frac{\phi}{|\sigma|} \cdot \alpha = \frac{2^{25}}{2^{65.7}} \cdot 2^{56} = 2^{15.3} \approx 40,000
\end{equation}
Thus, in case of uniform bin densities the attacker can compromise only $0.13\%$ of the password database (40,000 user accounts).
\begin{figure}[h]
$\begin{array}{ll}
\subfloat[Success rate of attacker $A$ using bin densities learned from Rockyou.]{\centerline{\includegraphics[width=0.6\textwidth]{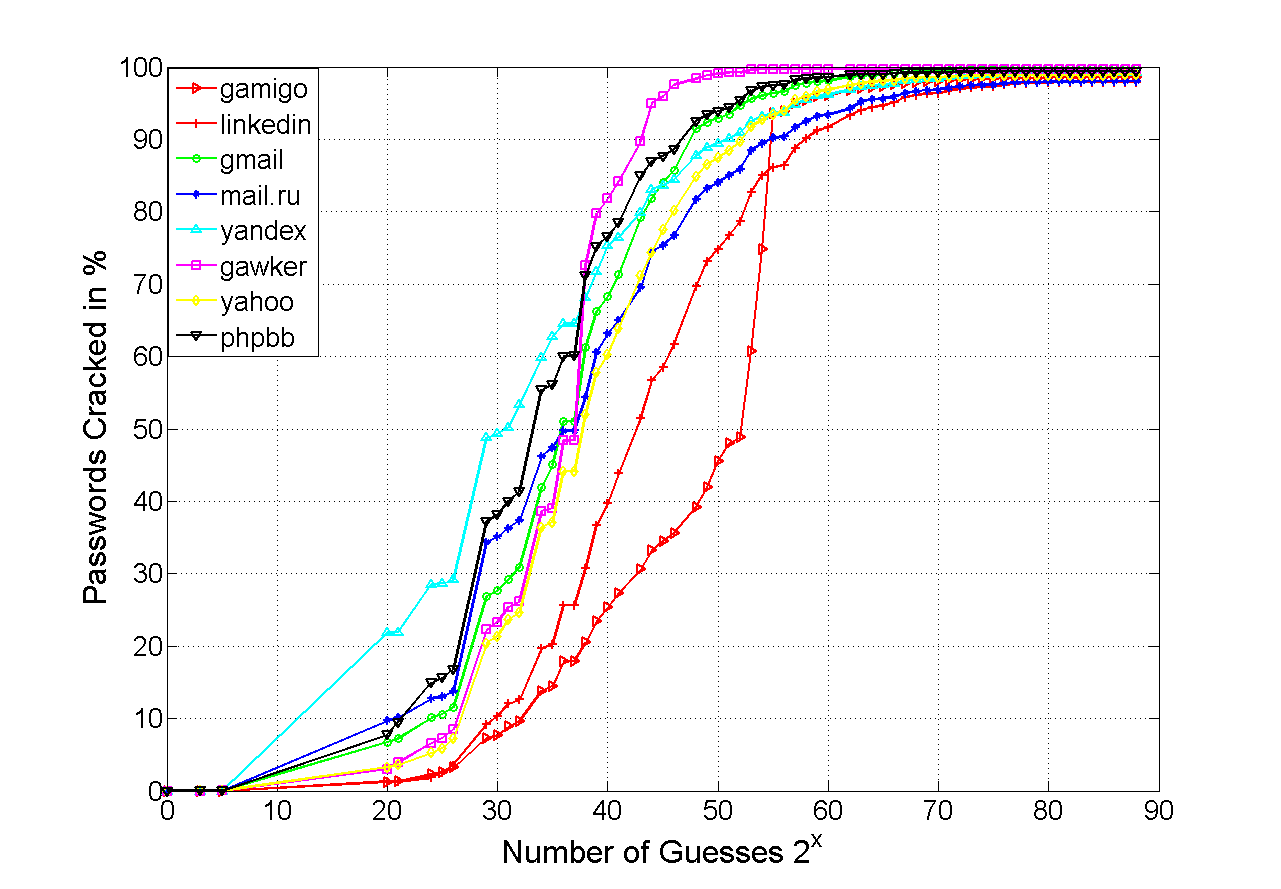}}} 
\\
\subfloat[Success rate of attacker $B$ using bin probabilities learned from Rockyou.]{\centerline{\includegraphics[width=0.6\textwidth]{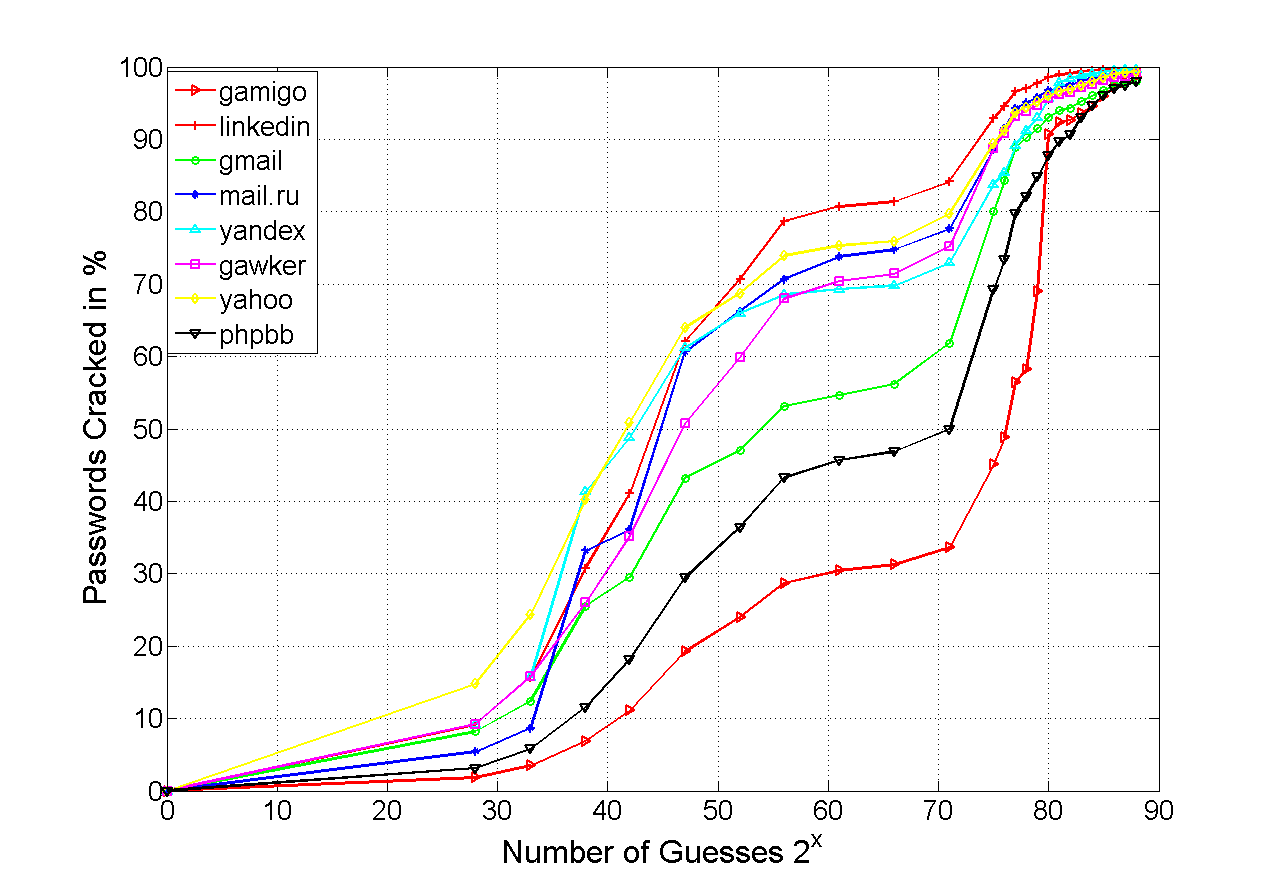}}}
\end{array}$
\caption{Comparing the success rate of attacker $A$ exploring denser bins with the success rate of attacker $B$ exploring popular bins.}
\end{figure}
However, the analysis of $32\ million$ Rockyou passwords reveals that the bin densities are highly non-uniform and with $\alpha = 2^{56}$ guesses, the bin attacker can recover nearly 94.34\% of the Rockyou database ($30.76\ million$ passwords). 
Now, we measure the resistance of the remaining 8 password databases against the bin attacker who learns the bin densities from the Rockyou database. The guessing efficiency of such bin attacker is depicted in Figure 6a and one can see that with $\alpha = 2^{56}$ guesses, the bin attacker can crack nearly 90\% passwords in every breached database. This implies that the bin densities derived from the Rockyou database closely approximates the bin densities of the remaining 8 password databases. 
\begin{table}[h]
\scriptsize
\centering
\caption{Percentage \% of passwords in the top dense bins of breached databases. These bins are smaller with length $l \le 8$.}~\label{tab:dense}
\begin{tabular}{|c|c|c|c|} \hline
Database &	$L_{\{1,8\}}$  & $D_{\{1,8\}}$  & $L_{x}D_{8-x}$ \\ \hline
Rockyou  & 26.80 & 12.19 & 17.88 \\ 
Gamigo & 6.28 & 3.13 & 8.74\\  
Linkedin  &  10.26	 & 2.71 &	15.39\\ 
Gmail	   &  27.80 & 11.72 &	17.61 \\ 
Mail.ru   &  17.32	 & 17.58 &	9.86\\
Yandex  & 11.84	& 40.03 &	6.56\\
Gawker  &  42.33 & 5.22   &	25.32\\ 
Yahoo	  &  24.44	& 5.14 &	20.78\\ 
Phpbb	  &  41.47	& 11.00 &	14.41\\  
\hline
\end{tabular}
\end{table}
Further analysis of breached databases reveals that smaller bins ($l \le 8$) of the form $D_{+}$, $L_{+}$ and $L_{+}D_{+}$ are much denser (Table~\ref{tab:dense}). For instance, there are nearly 26.80\% ($8.74 \ million$) Rockyou users in $L_{\{1,8\}}$ bins, 12.19\% ($3.97 \ million$) users in $D_{\{1,8\}}$ bins and 17.88\% ($5.83 \ million$) users in $L_{x}D_{8-x}$ bins, where $1\le x\le 7$. Similar behavior is observed in the other breached databases as well. The theoretical password space is huge ($95^l$, for 95 printable ASCII characters and password length $l$) but the utilized search space is very small. As a result, even with the small computing power $\alpha = 2^{40}$, the bin attacker (Figure 6a) can break nearly 60\% passwords in the breached databases\footnote{We did not get the count of passwords that belong to LinkedIn and Gamigo databases and consequently the success of the bin attacker with $\alpha = 2^{40}$ guesses is less than 60\% for both these databases.}. 
\subsection{Dense Bins vs Popular Bins}
In $Lemma \ 2$ we proved that the attacker benefits if the partitions are explored in decreasing order of density rather than in decreasing order of probability. Figure 6 shows the success rate of both these attack strategies for bin partitions. The attacker $A$ explores bins in decreasing order of density (Figure 6a) while the attacker $B$ explores bins in decreasing order of probability (Figure 6b). One can observe that even with the small computing power ($\alpha \le 2^{40}$) the attacker $A$ can gain huge success by exploring the denser bins. However, the expected success of the attacker $B$ who explores the most probable bins is insignificant. This is because the most probable bins are not necessarily the denser ones.  For instance, the bin $L_{8}$ with $2.46 \ million$ passwords is the third most popular bin and the bin $D_{6}$ with $2.28 \ million$ passwords is the fourth most popular bin in the Rockyou database. However, the density of $L_{8}$ bin is $\frac{2.46 \ million}{26^{8}} \approx 0.00001183$ while that of $D_{6}$ bin is $\frac{2.28 \ million}{10^{6}} \approx 2.28$. The attacker with the computing power say $\alpha = 2^{35}$ targeting the popular bins will exhaust all its resources without completely exploring the popular $L_{8}$ bin ($26^{8} \approx 2^{38}$ capacity) and recovering less than $2.46 \ million$ passwords. On the other hand the attacker targeting denser bins can recover $2.28 \ million$ passwords from $D_{6}$ bin and use the remaining $2^{35} - 2^{20} \approx 2^{35}$ computational power to explore the next denser bins. Therefore, the attacker exploring denser bins breaks more passwords by spending its resources wisely.
\\\\
{\em \textbf {Remark 1 :} In the absence of password composition policies, the resulting bin densities become highly non-uniform. As a result, the effort of the bin attacker is reduced drastically as a substantial fraction of the target password database can be recovered by exploring only the denser bins.}
\subsection{Composition Policy }
The composition rules forbid the use of bins composed of only one alphabet and enforce the use of bins composed of at least 2 or 3 alphabets. In this case, the search space is too large (at least $26 \cdot 26 \cdot 95^{l-2} > 95^{l-1}$) to carry out the brute-force search. However, merely increasing the search space does not imply that all the available bins are used uniformly. Surveys~\cite{Shay1,Anupam} and the analysis of breached databases suggest that  bins of the form $\{S,D\}_{p}U_{1}L_{n-2p-1}\{S,D\}_p$, where $p > 0$, that begin and end with digits or symbols and containing an uppercase letter followed by lowercase letters are more popular. To gain a better idea, we analysed the passwords composed of at least two alphabets from the breached databases. Analysis of these passwords provide more insights into bins that becomes highly dense as a result of enforcing different composition rules.
\begin{table}[h]
\tiny
\centering
\caption{Identifying potential denser bins used  in password creation when different alphabet sets are enforced upon users. The data (in percentage\%) indicates the popularity of denser bins for the given alphabet set. `-'  indicates the absence of corresponding alphabet set in the breached database.}~\label{tab:rules}
\begin{tabular}{|p{1.1cm}|p{2.2cm}|p{0.9cm}|p{0.9cm}|p{0.9cm}|p{0.9cm}|p{0.9cm}|p{0.9cm}|p{0.9cm}|p{0.9cm}|}
\hline
Set  & Dense Bins & Rockyou & Gamigo & LinkedIn & Gmail & Mail.ru & Yandex & Gawker & Yahoo \\ \hline 
$\{U,L\}$	&    $U^{1}L^{+}$           &71.43 & 31.19 &49.67  &-   &24.53   &62.70  &67.76  &57.17  \\ 
$\{U,D\}$	&    $U^{+}D^{+}$          &82.61 & 57.94 &61.85  &-   &43.85   &46.17  &81.14  &72.10 \\  
$\{L,D\}$	&    $L^{+}D^{+}$          &83.46  & 30.03 &63.92  &73.68  &46.07   &54.72  &58.36  &75.63  \\ 
$\{U,L,D\}$ &  $U^{1}L^{+}D^{+}$    &62.66  & 23.45 &41.66  &-   &28.30   &39.91  &42.92  &45.01  \\ 
$\{U,S\}$	&   $ U^{*}S^{1}U^{*}$  &69.07  & 43.63 &52.65  &-   &62.50   &67.15  &70.98  &57.14  \\ 
$\{L,S\}$	&    $L^{*}S^{1}L^{*}$	  &76.39 & 76.51 &73.36  &76.63  &83.30   &85.39  &85.96  &67.31  \\ 
$\{U,L,S\}$ &  $U^{[0,1]}L^{*}S^{1}U^{[0,1]}L^{*}$ &49.25   &61.62 &44.87   &53.66  &-   &53.35  &0.00   &36.09  \\ 
$\{D,S\}$	&   $ D^{*}S^{1}D^{*}$	   &45.43  & 27.00 &49.75  &59.46  &53.63   &50.93  &48.77  &40.00  \\ 
$\{U,D,S\}$ &  $U^{*}S^{1}U^{*}D^{+}$      &42.07 & 39.87  &39.50  &-   &41.79   &33.97  &55.76  &34.08  \\ 
$\{L,D,S\}$ &  $L^{*}S^{1}L^{*}D^{+}$       &43.13 & 46.57 &39.49  &42.02  &55.44  &47.26  &43.80  &34.39  \\ 
$\{U,L,D,S\}$& $U^{1}L^{+}D^{*}S^{1}D^{*}$     &30.85 & 38.36 &39.60  &-  &30.94  &31.26  &35.81  &25.56  \\
\hline\end{tabular}
\end{table}

It is evident from Table~\ref{tab:rules} that enforcing composition rules can still result in the non-uniform bin densities. For instance, the analysis of passwords derived using the alphabet set $\{U,L,D\}$ revealed that most of these passwords (62.66\% in Rockyou) are created using bins of the form $U_{1}L_{+}D_{+}$ that begin with an uppercase letter followed  by lowercase letters and end with digits. If the bin $U_1L_6D_2$ becomes denser, then even the attacker with $\alpha = 2^{40}$ can explore it completely as the bin size is $26^7\cdot 10^2 < 2^{40}$. Thus, division of the search space into bins can be very useful to crack the complex passwords which emerge due to composition policies. 
\\\\
{\em \textbf {Remark 2 :} The passwords created in the presence of composition rules can still result in the non-uniform bin densities. For instance, if users are enforced to create alpha-numeric passwords, the analysis of breached databases suggests that the bins of the form $U_{1}L_{+}D_{+}$ could become more dense. This observation can be easily exploited by the bin attacker thus defeating the purpose of the enforced composition rule.}
\section{Countermeasure}
In section 2, we proved that the partition attackers gains maximum benefit if the partitions are explored in decreasing order of density.
Now, we show that partition attackers can be countered only if the partition densities are made uniform. Before deriving the exact formula for the minimum expected success of a partition attacker, we prove few lemmas.
\begin{Lemma}
Let $0\le f \le 1$ and
\begin{equation}
\frac{\phi_{1}}{|\sigma_{1}|} \ge \frac{\phi_{2}}{|\sigma_{2}|} \ge \ldots \ge\frac{\phi_{n}}{|\sigma_{n}|}
\end{equation}
then for $1 \le j < n$,
\begin{equation}
\frac{\sum\limits_{i=1}^{j}\phi_i + f\cdot \phi_{j+1}}{\sum\limits_{i=1}^{j}|\sigma_i| + f\cdot |\sigma_{j+1}|} \ge \frac{\sum\limits_{i=1}^{j+1}\phi_i }{\sum\limits_{i=1}^{j+1}|\sigma_i|} 
\end{equation}
\end{Lemma}
\begin{proof}
\begin{equation}
(\sum\limits_{i=1}^{j}\phi_i + f\cdot \phi_{j+1}) \cdot (\sum\limits_{i=1}^{j+1}|\sigma_i|) -  (\sum\limits_{i=1}^{j+1}\phi_i)  \cdot (\sum\limits_{i=1}^{j}|\sigma_i| + f\cdot |\sigma_{j+1}|)
\end{equation}
\begin{equation}
\begin{split}
= \sum\limits_{i=1}^{j}\phi_i  \cdot \sum\limits_{i=1}^{j}|\sigma_i| +  |\sigma_{j+1}| \cdot  \sum\limits_{i=1}^{j}\phi_i  + f\cdot \phi_{j+1} \cdot \sum\limits_{i=1}^{j+1}|\sigma_i| \\ - \sum\limits_{i=1}^{j}\phi_i  \cdot \sum\limits_{i=1}^{j}|\sigma_i| - \phi_{j+1} \cdot \sum\limits_{i=1}^{j}|\sigma_i| - f\cdot |\sigma_{j+1}| \cdot \sum\limits_{i=1}^{j+1}\phi_i
\end{split}
\end{equation}
\begin{equation}
\begin{split}
= |\sigma_{j+1}| \cdot  \sum\limits_{i=1}^{j}\phi_i  -  f\cdot |\sigma_{j+1}| \cdot \sum\limits_{i=1}^{j}\phi_i  -  f\cdot |\sigma_{j+1}| \cdot \phi_{j+1}\\ - \phi_{j+1} \cdot \sum\limits_{i=1}^{j}|\sigma_i| + f\cdot \phi_{j+1} \cdot \sum\limits_{i=1}^{j}|\sigma_i| + f \cdot |\sigma_{j+1}| \cdot \phi_{j+1}
\end{split}
\end{equation}
\begin{align}
&= (|\sigma_{j+1}| \cdot  \sum\limits_{i=1}^{j}\phi_i ) \cdot (1 -  f)    - (\phi_{j+1} \cdot \sum\limits_{i=1}^{j}|\sigma_i|) \cdot (1 - f)\\
&= ( \sum\limits_{i=1}^{j}|\sigma_{j+1}| \cdot \phi_i -  \sum\limits_{i=1}^{j}\phi_{j+1} \cdot |\sigma_i| ) \cdot (1 -  f)\\
&= (1 -  f) \cdot \sum\limits_{i=1}^{j}(|\sigma_{j+1}| \cdot \phi_i -  \phi_{j+1} \cdot |\sigma_i| ) \ge 0
\end{align}
since $\frac{\phi_{i}}{|\sigma_{i}|} \ge \frac{\phi_{j+1}}{|\sigma_{j+1}|}$ for $1\le i \le j$. Therefore,
\begin{equation}
\frac{\sum\limits_{i=1}^{j}\phi_i + f\cdot \phi_{j+1}}{\sum\limits_{i=1}^{j}|\sigma_i| + f\cdot |\sigma_{j+1}|} \ge \frac{\sum\limits_{i=1}^{j+1}\phi_i }{\sum\limits_{i=1}^{j+1}|\sigma_i|}
\end{equation}
\end{proof}

\newtheorem{Corollary}{Corollary}
\begin{Corollary}
For $1 \le j \le n-1$,
\begin{equation}
\frac{\sum\limits_{i=1}^{j}\phi_i}{\sum\limits_{i=1}^{j}|\sigma_i|} \ge \frac{\sum\limits_{i=1}^{j+1}\phi_i }{\sum\limits_{i=1}^{j+1}|\sigma_i|}
\end{equation}
\end{Corollary}
\begin{proof}
The proof follows immediately by putting $f = 0$ in $Lemma \ 3$.
\end{proof}

\begin{Corollary}
For $1 \le j \le n$,
\begin{equation}
\frac{\sum\limits_{i=1}^{j}\phi_i}{\sum\limits_{i=1}^{j}|\sigma_i|} \ge \frac{\phi }{|\sigma|}
\end{equation}
where $\phi = \sum\limits_{i=1}^{n}\phi_i$ and $|\sigma| = \sum\limits_{i=1}^{n}|\sigma_i|$.
\end{Corollary}
\begin{proof}
The proof follows immediately from $Corollary \ 1$.
\end{proof}

\begin{Theorem}
The expected success of an attacker with limited computational power $\alpha < |\sigma|$ is minimum if the partition densities are uniform, {\em i.e.} $\frac{\phi_{1}}{|\sigma_{1}|} = \frac{\phi_{2}}{|\sigma_{2}|} = \ldots = \frac{\phi_{n}}{|\sigma_{n}|}$.
\end{Theorem}
\begin{proof}
 By $Lemma \ 1$, if the partition densities are uniform then the expected success of the attacker with computational power $\alpha < |\sigma|$ is,
\begin{equation}\label{eq:uni}
E_{uniform}(success) =  \frac{\phi}{|\sigma|} \cdot \alpha
\end{equation}
Now, consider the case of non-uniform partition densities. Without loss of generality assume that,
\begin{equation}
\frac{\phi^{'}_{1}}{|\sigma_{1}|} \ge \frac{\phi^{'}_{2}}{|\sigma_{2}|} \ge \ldots \ge \frac{\phi^{'}_{n}}{|\sigma_{n}|}
\end{equation}
where $\sum\limits_{i=1}^{n}\phi^{'}_i = \phi$.
\\
By $Theorem \ 1$, the maximum success for the attacker with computational power $\alpha < |\sigma|$ is,
\begin{equation}
E_{max}(success) =  \sum\limits_{i=1}^{i_{0}} \phi^{'}_{i}  + (\alpha -  \sum\limits_{i=1}^{i_{0}} |\sigma_{i}|) \cdot \frac {\phi^{'}_{i_0+1}}{|\sigma_{i_0+1}|}
\end{equation}
where $i_{0}$ is such that $ \sum\limits_{i=1}^{i_{0}}{|\sigma_{i}|} \le \alpha < \sum\limits_{i=1}^{i_{0}+1}{|\sigma_{i}|}$.
\\
Let $f = \frac{(\alpha -  \sum\limits_{i=1}^{i_{0}}|\sigma_{i}|)}{|\sigma_{i_0+1}|}$. Note that $0 \le f \le 1$.
\begin{equation}
E_{max}(success) =  \sum\limits_{i=1}^{i_{0}} \phi^{'}_{i}  + f \cdot \phi^{'}_{i_0+1}
\end{equation}
Using $Lemma \ 3$ we get,
\begin{equation}
E_{max}(success) \ge \frac{\sum\limits_{i=1}^{i_{0}+1} \phi^{'}_{i}}{\sum\limits_{i=1}^{i_{0}+1} \sigma_{i}}\cdot (\sum\limits_{i=1}^{i_{0}} \sigma_{i}  + f \cdot \sigma_{i_0+1})
\end{equation}
since $\sum\limits_{i=1}^{i_{0}} \sigma_{i}  + f \cdot \sigma_{i_0+1} = \alpha$ we have,
\begin{equation}
E_{max}(success) \ge \frac{\sum\limits_{i=1}^{i_{0}+1} \phi^{'}_{i}}{\sum\limits_{i=1}^{i_{0}+1} \sigma_{i}}\cdot \alpha
\end{equation}
By $Corollary \ 2$, we know that $\frac{\sum\limits_{i=1}^{i_{0}+1} \phi^{'}_{i}}{\sum\limits_{i=1}^{i_{0}+1}} \ge \frac{\phi}{|\sigma|}$
\begin{equation}
E_{max}(success) \ge \frac{\phi}{|\sigma|} \cdot \alpha = E_{uniform}(success)
\end{equation}
\end{proof}
\noindent
This result is more general and can be used to counter different instances of a partition attacker. However, we concern ourselves with the bin partitioning instance and in the light of $Theorem\ 2$, we propose a {\em bin explorer system} to counter the {\em bin attacker}.
\section{Bin Explorer System}
As demonstrated earlier, the bin attacker exploits the prevalence of non-uniform bin densities in the real-world password data and targets the denser bins to recover a major fraction of the password database. By $Theorem \ 2$, the success of the bin attacker can be minimized by making all bins equally dense. Now, we use this result to propose a new scheme to counter the bin attacker. For this purpose, consider the generation of a system assigned random password. It can be viewed as a sequence of following 3 steps.
\begin{enumerate}
\item Select the length $l$, {\em e.g.,} $l = 10$.
\item Randomly select the $l$ length bin $\beta$ from the collection of $n$ bins, {\em e.g.,} $\beta = L^4D^2L^2S^2$.
\item Randomly select a word in the bin $\beta$, {\em e.g.,} kebz93ga-?
\end{enumerate}
In such scheme, users have no control over the creation of their own passwords and the resulting passwords are also difficult to remember~\cite{Venkat,ShayPhrase}. To counter the bin attacker, we need to ensure only the uniform bin densities which can be achieved even if the system decides upon the first two steps of the random password creation process. Thus, the system can randomly assign the $l$ length bin $\beta$ to the user and allow the user to select a password from this assigned bin $\beta$. Therefore, we can mitigate the threat of the bin attacker and still provide the users with some control over their password creation. We refer to such system as {\em bin explorer}. 

From equation (\ref{eq:uni}) we observe that, as the number of users $\phi$ in the system increases, the expected success $E_{uniform}(success)$ of the attacker also increases. Therefore, just spreading users across different bins is not enough. The system should also ensure that the expected success $E_{uniform}(success)$ of the attacker is bounded. This can be achieved by increasing the search space size $|\sigma|$.
\subsection{Determining the Minimum Password Length}
 Again consider the settings where the maximum password length $l_{max}$ is 10, the attack search space is $|\sigma| \approx 2^{65.7}$, the computing power of bin attacker is $\alpha=2^{56}$ but the number of users $\phi = 2^{30}$. Assuming uniform  password bin densities, by equation (\ref{eq:uni}), the bin attacker can now compromise at least $E_{uniform}(success) = 2^{20.3} \approx 1.3\ million$ user accounts. Earlier the attack on a system with $\phi = 2^{25}$ users could recover only 40,000 passwords (equation (\ref{eq:32m})) but now the same attack when mounted on a system with $\phi = 2^{30}$ users looks more dangerous. The expected success $E_{uniform}(success)$ of the bin attacker is directly proportional to the number of users $\phi$ in the system. More the number of users in the system, more the expected success of the bin attacker. Thus, the size of the search space $|\sigma|$ should also depend upon the number of users $\phi$ in the system. Rearranging the terms in the equation (\ref{eq:uni}) we get,
\begin{equation}
|\sigma| = \frac{\phi \cdot \alpha}{E_{uniform}(success)}
\end{equation}
Using (\ref{eq:ss}) we get,
\begin{equation}
|\gamma|^l = \frac{\phi \cdot \alpha}{E_{uniform}(success)}
\end{equation}
where, $\gamma$ is the alphabet set from which the password is derived and $l$ is the password length. Using this relation, one can precisely compute the minimum password length $l_{min}$ as, 
\begin{equation}\label{eq:lmin}
l_{min} = \frac{log(\phi)+log(\alpha)-log(E_{uniform}(Success))}{log(|\gamma|)}
\end{equation}
Now, we can use this result to decide the minimum password length and thus restrict the bin attacker with the computational power $\alpha$ to system-desired success $E_{uniform}(success)$.
\begin{figure}[h]
\centering
\includegraphics[scale=0.7]{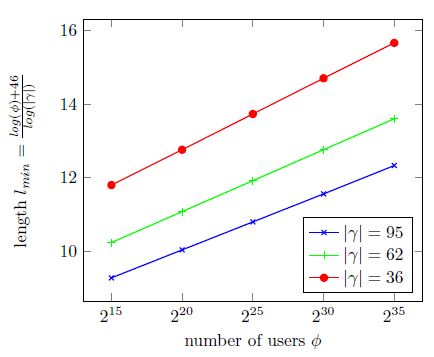}
\caption{Determining minimum password length $l_{min}$ for different values of $\phi$ and alphabet set $\gamma$. The parameters depicting attackers guessing capability $\alpha=2^{56}$ and expected success $E_{uniform}(success) = 2^{10}$ are fixed.}~\label{fig:binlen}
\end{figure}

Figure~\ref{fig:binlen} depicts the minimum password length $l_{min}$ required to bound the expected success $E_{uniform}(success) = 2^{10}$ of the bin attacker with the computational resource $\alpha=2^{56}$. The password length is determined for different systems which vary in both the number of users and the alphabet sets used for the bin generation. For instance, if the number of users $\phi = 2^{25} = 32 \ million$ and the alphabet size $|\gamma| = |L|+|U|+|D|+|S| = 95$, then the minimum password length should be 11. The number of available bins in this case is $n\approx 4^{11}$ and the density of every bin is $\frac{\phi}{|\sigma|} = 2^{25}/95^{11} = 2^{25}/2^{72.27} \approx 2^{-47.3}$. Hence, because of the uniform bin densities, the effort of the offline attacker is increased exponentially.
\subsection{User-Bin Assignments}
To protect against offline attacks, websites enforce password composition policy on users. However, the enforced policy is same for all users of a given website which can again result in a fewer denser bins and a large number of empty bins. To prevent the underutilization of a search space, the system should play an active role in distributing users across different bins. A few distribution strategies are as follows (comparison in Table~\ref{tab:cmpstr}).
\\\\
\textbf{Round Robin.} In this strategy, the system assigns a bin to a new user in a round-robin fashion. For this purpose, the system maintains an array of n bins along with an index $b < n$ which points to the next available bin in the array. After the arrival of a new user, the next available bin $\beta_b$ is allocated and the value of the index is increased to $(b+1)\ mod\ n$. However, as a result of this strategy, bins of various sizes have the same number of passwords $\frac{\phi}{n}$ which leads to non-uniform bin densities. Further, assigning bins in a fixed order opens up an attack wherein the attacker who gets access to the system and can infer the order in which users created accounts now knows the structure of a particular user's password. The attacker can therefore mount targeted attacks on a particular user account, e.g., admin. The space required for maintaining the data-structure is proportional to the total number of bins n and the time required to allocate a bin is $O(1)$.
\\\\
\textbf{Density-ordered.} In this allocation strategy, the system strives to make every bin equally dense by assigning the least dense bin to a new user. To achieve this, the system stores all bins in a list in decreasing order of density. When a new user arrives, the system selects the least dense bin $\beta_{min}$ from the end of this sorted list. If there are multiple least dense bins, the system chooses a bin uniformly at random, thereby increasing the effort of targeted attacks, e.g., if the density of all bins composed of 6 lowercase letters and 2 digits is equal, then the system chooses one of these ${8 \choose 2}$  bins uniformly at random. Now, in order to succeed the targeted attack requires searching all ${8 \choose 2}$  bins.

As the list is already sorted, one can use binary search technique to identify a region within the list containing bins with least density. Subsequently, the system randomly picks a bin $\beta_{min}$ from this region and allocates it to the user. Finally, the density of the allocated bin $\beta_{min}$ is updated and its position in the list is rectified using the binary search algorithm. The space required for maintaining the sorted list is proportional to the total number of bins $n$ while the time require to allocate the least dense bin is $O(logn)$.
\\\\
\textbf{Random.} In this strategy, the system randomly generates a bin for every user. The probability of assigning the bin should be proportional to its size. Otherwise the bins of different capacities will have the same number of passwords, disturbing the density uniformity. Thus, the expected number of passwords in any bin $i$ is $\phi_i = \phi \cdot \frac{ |\sigma_i| }{|\sigma|}$. Therefore, the expected density of a bin is $\frac{\phi}{|\sigma|}$. Since the bins are assigned randomly some bins can get $\frac{\phi}{|\sigma|} \cdot x - \frac{\phi}{|\sigma|}$(where $x \ge 1$) more users than the expected. We call the quantity $x$ as stretch. For the random assignment strategy the stretch $x$ is not more than $\frac{logn}{loglogn}$ with high probability~\cite{Gonnet}. Since $logn = log(4^{l_{min}}) = 2\cdot l_{min}$, $x = \frac{2\cdot l_{min}}{log(2\cdot l_{min})}$.
In this strategy, there is no need to keep the track of number of users in each bins. Also no pointers are required. The bins are generated randomly upon the arrival of the new user and therefore, the required space and time is just $O(1)$. Further, random bin assignment makes targeted attacks on a particular user arduous. Therefore, the random approach is better among all proposed approaches.
\\\\
\textbf{Power of Two Choices.} In this strategy, the system randomly chooses two bins and assigns the least dense bin to the new user. Again the probability of choosing any bin should be proportional to its size. Otherwise the bins of different capacities will have same number of passwords, disturbing the density uniformity. Such implementation would require the storage to remember the number of users in every bin $\sigma_i$. Therefore, $O(n)$ space is required. The expected density of a bin is $\frac{\phi}{|\sigma|}$. However, the stretch in this strategy depend upon the size of a bin~\cite{nonuniform}. In case of smaller bins~\cite{nonuniform} the stretch is $2 \cdot loglogn = 2 \cdot log(2 \cdot l_{min})$ and for the larger bins~\cite{nonuniform} the stretch is $4 + \varepsilon$ with high probability.
In case of database compromise, the attacker can learn the exact number of users in any given bin.
\begin{table}[h]
\scriptsize
\centering
\caption{Comparison of different bin distribution strategies.}~\label{tab:cmpstr}
\begin{tabular}{|p{2cm}|p{0.7cm}|p{1cm}|p{1.2cm}|p{2cm}|} \hline
Strategy& Space & Time & Expected density & Stretch $x$\\ \hline
Round Robin	& $O(n)$ & $O(1)$ & $\frac{\phi}{n\cdot |\sigma_i|}$ & $\frac{\phi}{n\cdot |\sigma_i|} - \frac{\phi}{|\sigma|}$\\ \hline
Density-ordered & $O(n)$ & O(logn) & $\frac{\phi}{|\sigma|}$ & $1+\varepsilon$\\ \hline
Random         & $O(1)$  & $O(1)$ & $\frac{\phi}{|\sigma|}$ & $\frac{2\cdot l_{min}}{log(2\cdot l_{min})}$\\ \hline
Power of Two Choices & $O(n)$ & $O(1)$  & $\frac{\phi}{|\sigma|}$ & $2 \cdot log(2\cdot l_{min})$ for smaller bins and $4+\varepsilon$ for larger bins\\ \hline
\end{tabular}
\end{table}
\subsection{User-Bin Adaptation}
Various techniques have been proposed in the past to help users create secure and memorable passwords~\cite{Yan,Jeya,Sonia,BonneauSecret}. Now, we show how these existing schemes can be easily adapted for helping users in creating passwords from the system assigned bin.
\\\\
\textbf{Non-Iterative Scheme.} This scheme directly asks users to create a password according to a system assigned bin in one step.
\begin{enumerate}
\item System chooses a random bin $\beta$ of length $l \ge l_{min}$ from the collection of predefined bins, {\em e.g.,} $\beta = U^{1}S^{1}L^{1}D^{2}S^{1}U^{1}L^{2}$
\item User creates a password according to the system assigned bin, {\em e.g.,} \\ ``D@c45\&Mac".
\end{enumerate}
To help users, one can also modify the system proposed in~\cite{Jeya} to create mnemonic passphrases corresponding to the system assigned bin, {\em e.g.,} for bin \\$\beta = U^{1}S^{1}L^{1}D^{2}S^{1}U^{1}L^{2}$ the mnemonic passphrase can be ``It's 12 noon I am hungry'' and therefore the password is ``I's12\&Iah''.
\\\\
\textbf{Iterative Scheme.} In 2009, Forget {\em et al.}~\cite{Sonia} showed that the security of the text passwords can be improved either by inserting random characters in the user chosen password or by randomly replacing the characters in the user chosen password. While, Bonneau {\em et al.}~\cite{BonneauSecret} proposed an incremental approach to imprint 56-bit secret into human memory. Both these methods can be combined to help users to choose a secret from the assigned bin. The scheme is as follows,
\begin{enumerate}
\item User chooses a lowercase password of length $l \ge l_{min}$ subject to a blacklist of say top 10,000 popular passwords.
\item System chooses a random bin $\beta$ of length $l_{min}$ from the collection of predefined bins. 
\item System calculates the hamming distance $hd$ between the user chosen bin $L^{l_{min}}$ and system assigned bin $\beta$.
\item After every $x$ successful logins, the system attempts to minimize the hamming distance $hd$ by 1. User should replace the lowercase letter in the password by a letter corresponding to the bin $\beta$.
\item After $hd\cdot x$ number of successful logins, the password from bin $L^{l_{min}}$ is transformed to system assigned random bin $\beta$.
\end{enumerate}
The system proposed in~\cite{BonneauSecret} requires the assigned 56-bit secret to be stored in plaintext until user learns it. But in our scheme, the password is stored securely at every stage. The system only has to store the assigned bin in plaintext until the password is created as per this assigned bin. Once the final password is set, this bin information is also removed from the system.
\\\\
\textbf{Password Manager.} Password manager is one of the convenient ways for securely managing passwords. It generates, stores and recalls passwords on behalf of the users. Password managers such as~\cite{Bruce,Dominik} generate passwords based upon the password policy. However, password managers generate such passwords independently in isolation without communicating with the server. Such strategies can lead to maximum load of $logn/loglogn$ on few bins with high probability~\cite{Gonnet}. Uniform bin densities can be ensured if password managers generate passwords using the system assigned bin. The password manager's logic can be easily tweaked to create passwords according to the assigned bins.

Recent studies~\cite{Emperor,Autofill} revealed serious vulnerabilities in the implementation of different password managers. We believe that security conscious users can still use password managers to store the partial information instead of storing the entire password of critical accounts. This partial information can be used to recall the entire password whenever required. For instance, in the bin explorer schemes one can use password managers to store the password bin information as a hint and use it to recall their actual password during login.

There are also studies which indicate that the presence of strength meters can result in stronger passwords~\cite{Egelman,Shay4}. We can use this strategy in influencing user behavior towards rarer bins. When users choose any of the denser bins the strength meter can warn them about the risks of offline attacks.

\section{Experiments}
We studied the usability of both non-iterative and iterative schemes described in the previous section with 33 users. Both these studies were conducted with the same 33 users in a laboratory in the presence of an experimenter. All users were graduate working professionals and regular internet users. 10 users (30.3\%) were female and 23 users (69.7\%) were male belonging to the same nationality. 

We studied the usability of a non-iterative scheme followed by that of a iterative scheme. In both studies, we used a set of 10 length bins composed of exactly 6 lowercase letters and remaining 4 letters were derived using any 3 alphabets $U,D,S$ $e.g.,$ $LLDLDLUULL$, $LLLSLULDDL$ and so on. The number of 10 length bins composed of exactly 6 $L$ and any 3 alphabets $U,D,S$ in the remaining 4 positions is ${10 \choose 4} \cdot 3^{4} = 17010$. These bins were allocated using random distribution strategy and therefore each user is assigned a unique bin with a very high probability (nearly 0.97). We asked users to use the assigned 10 length bin to create a minimum 10 length password. The first 10 letters of the password had to follow the assigned bin pattern and the remaining letters of the password had no restriction. Now, we describe our experiments in detail.
\subsection{Non-Iterative Scheme}
The non-iterative experiment was conducted in two phases. We refer to the first phase as Password Creation phase and the second phase as Password Recall phase.
\begin{figure}[h]
\centering
\includegraphics[scale=0.6]{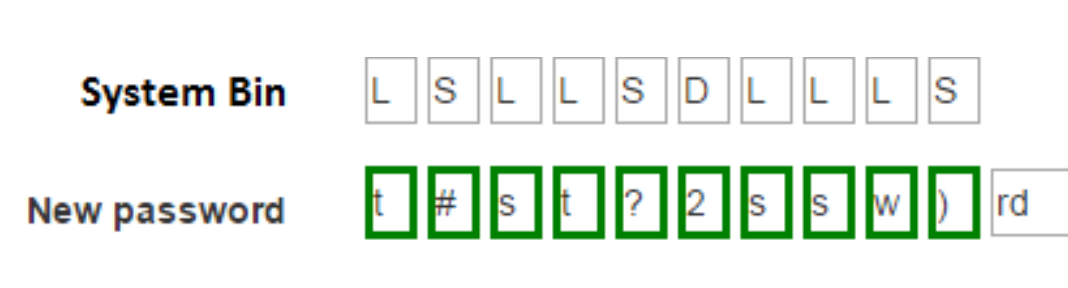}
\caption{Creating password using system assigned bin $LSLLSDLLLS$ in the non-iterative experiment.}~\label{fig:noniter}
\end{figure}
\begin{enumerate}
\item \textbf{Password Creation phase.} In this phase, users were asked to create a minimum 10 length password using the system assigned bin displayed on the password creation page (Figure~\ref{fig:noniter}). After retyping password, users were directed to the login page, where they had to enter their password again. 
\item \textbf{Password Recall phase.} After 72 hours, we invited users for the recall phase. There were no practice sessions between these two stages. This setup helped us to measure the recall of the password, when password has not been used for a while. If users were not able to recall their password within at most 5 attempts, we showed them their bin which was used during the password creation. After displaying the bin, 3 more attempts were allowed. If the users were not able to recall their password even after displaying the bin, we showed them their password. Finally, we asked users to fill a short survey. The questions mostly captured the user sentiment and password storage behavior. These questions and responses are listed in~\cite{trddc}. 
\end{enumerate}
Now, we investigate the memorability and efficiency of the non-iterative scheme. 
\\\\
\textbf{Memorability.} We measure memorability in terms of: 
\begin{enumerate}
\item Number of users who successfully recalled their passwords during the password recall phase.
\item Average number of login attempts required for the successful recall. 
\end{enumerate}
For convenience, we classify users into 3 categories, namely {\em password storage}, {\em bin storage} and {\em no storage}. 3 users (9.10\%) who reported writing down their password belong to the {\em password storage} category, 4 users (12.12\%) reported writing down their bin belong to the {\em bin storage} category and remaining 26 users (78.78\%) belong to the {\em no storage} category. All users who wrote down their password and 3 users who wrote down their bin recalled their password successfully. 26 users who did not report storing their passwords or bins were further split into two categories, namely {\em no hint} and {\em bin hint}. 11 users from the {\em no storage} category who successfully recalled their passwords without requiring any hint belong to the {\em no hint} category. The remaining 15 users failed in the password recall and we helped them by displaying their corresponding bin which was used during the password creation. These 15 users belong to the {\em bin hint} category. Upon viewing the bin, 10 more users succeeded in recalling their passwords. Thus, overall 21 users from the {\em no storage} category recalled their passwords and 5 users forgot their password. The results are summarized in Table~\ref{tab:recall}.
\begin{table}[h]
\centering
\scriptsize
\caption{Data of users who successfully recalled their passwords in non-iterative study.}~\label{tab:recall}
\begin{tabular}{|p{2.5cm}|p{1cm}|p{1.5cm}|p{1.3cm}|p{1.3cm}|} \hline
Category & Total Users & Successful Recalls & Average Login Attempts & Average Login Time\\  \hline
Password Storage	& 3 & 3 (100\%) & 1.00 & 21.22s\\
Bin Storage          & 4  & 3 (75\%) & 1.33 & 35.52s\\ 
No Storage-No Hint        & 11 & 11 (100\%) & 1.72 & 37.94s\\ 
No Storage-Bin Hint       & 15 & 10 (66.67\%) & 3.73 & 77.20s\\ 
\hline
\end{tabular}
\end{table}

The average number of login attempts required to successfully login during the recall phase is also shown in Table~\ref{tab:recall}. The users who successfully recalled their passwords during the recall phase with or without {\em bin hint} are referred to as {\em successful users}. Users who wrote down their bins required 1.33 login attempts while users in {\em no hint category} required 1.72 attempts on an average. 
The users in the {\em bin hint} category made 2.53 attempts on an average before clicking on {\em Forgot Password} button. After viewing their bin, users required 1.20 more attempts to succeed. Therefore, average login attempts of users in {\em bin hint} category is much higher (3.73 attempts). 
\\\\
\textbf{Efficiency.} We measure efficiency in terms of:
\begin{enumerate}	
\item Average password creation time required during the creation phase
\item Average login time required during the recall phase.
\end{enumerate}
Password creation required nearly 139 seconds on an average. The average login time required for the successful recall of password during the recall phase is shown in Table~\ref{tab:recall}. The time required for recalling passwords in {\em password storage}, {\em bin storage} and {\em no hint} category is less compared to the time required for recalling passwords in {\em bin hint} category. This is because the average login attempts required for the successful {\em bin hint} category users is much higher (3.73 attempts).
\\
\\
\textbf{User Sentiment.} The memorability and efficiency results of the non-iterative experiment indicate that passwords created using randomly assigned bins are both difficult to remember and to enter. 
\begin{table}[h]
\centering
\scriptsize
\caption{User sentiment in non-iterative experiment.}~\label{tab:response}
\begin{tabular}{|p{2.0cm}|p{2.0cm}|p{2.0cm}|p{2.0cm}|} \hline
Likert Scale & Creating password was easy? & Remembering password was difficult? & Created password was more secure?\\ \hline
StronglyAgree	& 0 (0\%) & 11 (33.33\%) & 11 (33.33\%)\\
Agree         & 11 (33.33\%)  &  14 (42.42\%) &14 (42.42\%)\\ \hline
Neutral          & 8 (24.24\%)& 5 (15.15\%) &4 (12.12\%)\\ \hline
Disagree &8 (24.24\%) &2 (6.06 \%) &3 (9.09\%))\\
StronglyDisagree & 6 (18.18\%) &1 (3.03\%) &1 (3.03\%)\\
\hline
\end{tabular}
\end{table}
This can also be observed from the survey responses that we received from the users (Table~\ref{tab:response}). Since random bins require users to place digits, symbols and uppercase letters at 4 random positions, the resulting passwords are difficult to remember and to enter. However, most users (75.75\%) felt that resulting passwords are more secure than the passwords that they usually create.
\subsection{Iterative Scheme}
To help users in remembering digits, symbols and uppercase letters at random positions, we designed an iterative experiment. We performed this experiment with the same 33 users in 2 phases.
\begin{figure}[h]
\centering
\includegraphics[scale=0.6]{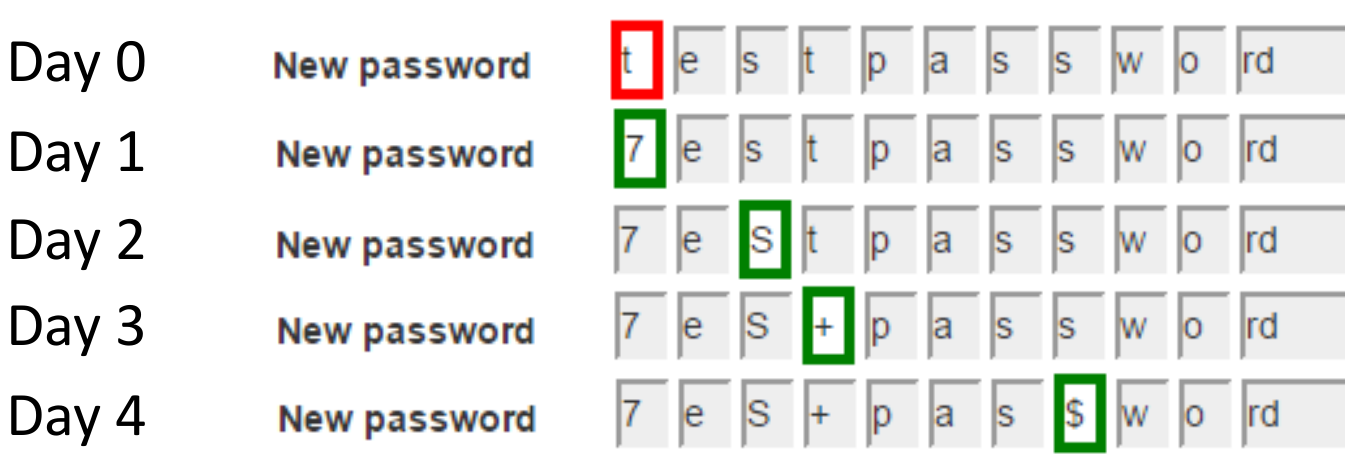}
\caption{Evolution of a simple password ``testpassword" to a complex password ``7eS+pas\$word" using system assigned bin ``DLUSLLLSLL". On Day1, system suggests replacement of letter `t' in highlighted red box with any digit $D$. This process continues for next 3 days until first 10 letters of ``testpassword" are aligned with 10 length bin ``DLUSLLLSLL".}~\label{fig:iter}
\end{figure}
\begin{enumerate}
\item \textbf{Password Creation phase. }This phase was split into 5 stages and required 5 days for completion. The stages were numbered from 0 to 4. In the $0^{th}$ stage, users were asked to create a minimum 10 length lowercase password ($L^{10}$). Further, the system generated a random 10 length bin for every user and stored it in the password database. This bin is not displayed to the user. In the next 4 stages, system guided users in transforming the first 10 letters of their password from $L^{10}$ bin to the system assigned bin. To complete the $i^{th}$ stage users were asked to login using the password created in the $(i-1)^{th}$ stage. Upon successful login, the system minimizes the hamming distance between the current bin and system assigned bin by suggesting a letter replacement in the current password with an appropriate alphabet from the set $\{U,L,D,S\}$. The user can choose any letter from the suggested alphabet and replace the letter of a current password as suggested by the system. Thus, after completing the $4^{th}$ stage, the final password of user differed from the password created in $0^{th}$ stage in 4 random positions and was according to the assigned bin (Figure~\ref{fig:iter}).
\item \textbf{Password Recall phase.} The password recall phase in the iterative experiment is same as defined in the non-iterative experiment.
\end{enumerate}
Now, we investigate the usability of the iterative scheme.
\\
\\
\textbf{Memorability.} Since users successfully recalled their passwords during different iterations of creation phase, we report the memorability results of only recall phase. In the iterative study, we classify users into 2 categories, {\em password storage} and {\em no storage}. We do not have bin storage category since no one reported writing their bin. The reason is that the bin was never displayed to the users directly during the password creation. 5 users (15.16\%) who reported writing their password belong to the {\em password storage} category and the remaining 28 users (84.84\%) belong to the {\em no storage} category. All users who wrote their passwords, successfully logged into the system. Moreover, 26 users in {\em no storage} category successfully recalled their password without any hint ({\em no hint}). Thus, the iterative system resulted in huge improvement compared to the previous non-iterative system. We helped 2 unsuccessful users in their password recall by showing them their corresponding bin which was used in the password creation. Upon viewing the bin, 1 more user succeeded in recalling the password. Thus, overall 27 users from {\em no storage} category recalled their passwords and only 1 user forgot the password. 
\begin{table}[h]
\centering
\scriptsize
\caption{Data of users who successfully recalled their passwords in iterative study.}~\label{tab:iterrecall}
\begin{tabular}{|p{2.5cm}|p{1cm}|p{1.5cm}|p{1.3cm}|p{1.3cm}|} \hline
Category & Total Users & Successful Recalls & Average Login Attempts & Average Login Time\\  \hline
Password Storage	& 5 & 5 (100\%) & 1 & 25.19s\\
No Storage-No Hint         & 26 & 26 (100\%) & 1.65 & 35.05s\\
No Storage-Bin Hint         & 2 & 1 (50\%) & 3 & 59.21s\\
\hline
\end{tabular}
\end{table}

The average login attempts required for successful login during the recall phase is shown in Table~\ref{tab:iterrecall}. Users who wrote down their passwords or bins recalled their passwords in just 1 attempt. Users in {\em no hint} category required just 1.65 login attempts on an average during the recall phase. Since users in the {\em bin hint} category were shown bin after few login attempts (at most 5), therefore their average login attempt is higher (3).
\\
\\
\textbf{Efficiency.} Password creation required nearly 182.3 seconds on an average. This creation time includes the time required for creating a lowercase password in the $0^{th}$ iteration and the time to replace a letter with an appropriate alphabet in the subsequent 4 iterations. The average login time required for the successful recall of password during the recall phase is shown in Table~\ref{tab:iterrecall}. The amount of time required for recalling passwords in the {\em password storage} and {\em no hint} category is less compared to the amount of time required for recalling passwords in the {\em bin hint} category.
\\
\\
\textbf{User Sentiment.} The memorability and efficiency results of the iterative study indicates that the passwords created using randomly assigned bin can be remembered. 
\begin{table}[h]
\centering
\scriptsize
\caption{User sentiment in iterative experiment.}~\label{tab:iterresponse}
\begin{tabular}{|p{2.0cm}|p{2.0cm}|p{2.0cm}|p{2.0cm}|} \hline
Likert Scale & Creating password was easy? & Remembering password was difficult? & Created password was more secure?\\ \hline
StronglyAgree	& 7 (21.21\%) & 4 (12.12\%) & 10 (30.30\%)\\
Agree                      & 13 (39.39\%)  &  13 (39.39\%) &18 (54.54\%)\\ \hline
Neutral                    & 9 (27.27\%) & 10 (30.30\%) &4 (12.12\%)\\ \hline
Disagree                 & 3 (9.09\%) &6 (18.18\%) &1 (3.03\%)\\
StronglyDisagree   & 1 (3.03\%) &0 (0\%) & 0 (0\%)\\
\hline
\end{tabular}
\end{table}
The survey responses also indicate that creating password using iterative method is much easier compared to non-iterative method (Table~\ref{tab:iterresponse}).
\subsection{Non-Iterative vs Iterative Scheme}
We highlight the prominent differences in the usability results of iterative and non-iterative schemes. 
\begin{enumerate}
\item \textbf{Creation time.} The average password creation time in the iterative scheme (182.3s) is more compared to the non-iterative scheme (139s) because of the stage-wise transformation of the password.
\item \textbf{Recall and Login attempts.} In the iterative scheme, 31 users (93.93\%) recalled their passwords successfully while in non-iterative scheme only 17 users (51.51\%) succeeded in their password recall. Moreover, the average login attempts for 31 successful users in the iterative scheme is just 1.54 while the average login attempts for 17 successful users in the non-iterative scheme is 1.89. This indicates that the users in the iterative scheme can recall their passwords more reliably without requiring any hint.
\item \textbf{User sentiment.} Only 4 users (12.12\%) found it difficult to create the password using the iterative scheme. On the other hand, 14 users (42.42\%) found it difficult to create the password using the non-iterative scheme. Thus, passwords in the iterative scheme are easy to create as well as easy to remember. Moreover, 28 users (84.84\%) in the iterative scheme and 25 (75.75\%) users in the non-iterative scheme felt that the resulting passwords are secure. Responses are listed in~\cite{trddc}.
\end{enumerate}

All 33 participants in our study were working professionals and may have better memory than average, which could positively influence the usability results. We only had 33 participants and with a larger population we might be able to observe further patterns. However, the purpose of this study was to observe whether the iterative or non-iterative scheme is usable for creating passwords from system assigned bins. As the usability data indicates, the iterative scheme fared better than the non-iterative in many aspects.
\section{Related Work}
\textbf{User behavior and Attacks.} A great deal of research effort has been spent in describing the password creation strategies of users and finding the best techniques to break these passwords in an offline mode. In 1978, Morris and Thompson~\cite{Morris} analysed 3289 user passwords and found that 86\% of these are weak either due to their prevalence in the dictionary or due to their short length. Further these passwords contained only lowercase letters or digits. Later in 1990, Klein~\cite{Klein} successfully broke 25\% of the passwords on the Unix system using the brute-force attack. In 1999, Moshe and William~\cite{Moshe} based upon their survey of 997 participants found that 80\% of the passwords were derived using only lowercase and uppercase letters. Then in 2005, Narayanan and Shmatikov~\cite{Narayanan} demonstrated the effectiveness of Markov models by generating the candidate passwords in decreasing order of the probability and cracking 67.6\% of the passwords. In 2007, Flor\^{e}ncio and Herley~\cite{Florencio} studied the passwords of nearly $5 \ million$ users and found that most of these passwords are composed using either lowercase letters or digits. In 2009, Weir {\em et al.}~\cite{Weir} proposed generating most probable guesses by learning probabilistic context-free grammar (PCFG) from the breached databases. In our work, we proved that the density-based guesses are more effective than the probability-based guesses. We demonstrated that the bin attacker gains maximum benefit if the bins are explored in decreasing order of density rather than in decreasing order of probability.

Mostly the password research is conducted by analysing the passwords in the publicly available breached databases~\cite{Weir,Shay3,Chinese} or by analysing the passwords collected using the surveys~\cite{Shay1,Shay2,Shay2014}. Bonneau however, for the first time performed the large scale study of anonymized $70 \ million$ passwords of Yahoo database in~\cite{BonneauMetric}. Bonneau also derived a useful measure for gauging the strength of a password database. However, computing this measure requires knowledge about the probability distribution of passwords. We relaxed this constraint and proposed a more general attack model which exploits the available information to divide the search space into partitions and explores them in decreasing order of density.

The presence of composition rules is believed to create secure passwords~\cite{Summers} while~\cite{Proctor,Shay2,Shay2014} suggest that longer passwords also provide equivalent security. However, all these results are based upon the survey of not more than a few thousand users and there is no real data available to study the passwords created in the presence of composition rules. We argue that the presence of composition policies can also be exploited by the attacker if the partition densities remain non-uniform.
\\
\textbf{Countermeasures.} To counter offline attacks, a significant amount of effort has been, and is being, invested to help users in remembering a high entropy secret. Various schemes have been proposed to influence the password creation strategies of users which includes the use of mnemonic-based passwords~\cite{Yan,Jeya}, system modified passwords~\cite{Sonia}, system assigned passwords~\cite{ShayPhrase}, long passwords~\cite{Shay2014} and pronounceable passwords~\cite{lau}. A recent study~\cite{BonneauSecret} shows that users can remember a randomly assigned 56-bit secret using the technique of spaced repetition while another recent study~\cite{Blum} helps user in generating a user friendly password which is difficult for machine to crack. {\em These studies emphasize the need and importance of high entropy passwords specifically to counter the offline attacker.} 

Motivated by these research, we proposed {\em iterative} scheme to help users in recalling their password from the system assigned bin. We purposefully adapted the system modified password scheme proposed by Forget {\em et al.}~\cite{Sonia} and the spaced-repetition technique proposed by Bonneau {\em et al.}~\cite{BonneauSecret} to make our iterative scheme more usable. Moreover, our adapted iterative scheme addresses usability concerns associated with the Forget {\em et al.} scheme~\cite{Sonia}  and security concerns of Bonneau {\em et al.} scheme~\cite{BonneauSecret}. 
\begin{itemize}
\item In 2008, Forget {\em et al.}~\cite{Sonia} showed that users accept few system suggested modifications to their password. The authors tested 4 conditions, namely, {\em Insert-2}, {\em Insert-3}, {\em Insert-4} and {\em Replace-2} with 16 participants in each condition. Also, all modifications were suggested to the user at once (non-iterative). Consequently, Insert-4 scheme which proposes 4 insertions to the user-chosen password were found to be unusable, even though the recall phase was conducted immediately after answering two questions and performing a distraction task of 45 sec. 
In our work, we tested {\em Replace-4} condition in both iterative and non-iterative fashion with 33 participants and showed that users are more likely to accept system suggested replacements when the approach is incremental. Further, we conducted the recall-phase after a delay of 72 hours from the password creation phase. Also, note that {\em Replace-4}  condition is new and has never been tested before.
\item The spaced repetition approach proposed by Bonneau {\em et al.}~\cite{BonneauSecret} requires the system-assigned password to be stored either in plaintext or in encrypted form until the user memorizes it. But in our iterative scheme, the password is stored securely using a hash function at every stage. The system only has to store the assigned bin in encrypted form until the password is created as per this assigned bin. Once the final password is set, this bin information is also removed from the system. If the server is compromised, the attacker has access to both the encryption key and the database. In this scenario, the spaced repetition scheme of Bonneau {\em et al.}~\cite{BonneauSecret} reveals the entire password while our scheme reveals only the bin information.
\end{itemize}

\section{Discussion}

\subsection{Hybrid Bin Attacker}
The main advantage of the bin attacker is that it does not require any dictionary to operate. The attacker takes advantage of the fact that users create passwords predominantly using lowercase letters or digits and if there are any uppercase letters and symbols in the passwords then they are at the predictable positions. Recently, researchers showed that human-generated passwords follow Zipf's law~\cite{Zipf}. The observation is that few passwords are very frequent and constitute a good fraction of the password database. For instance, the analysis of Rockyou database~\cite{Rockyou1} reveals that 5000 most popular passwords were used by almost 20\% of the users ($6.4 \ million$ users) while recently released password frequency list of a live Yahoo database~\cite{Yahoo} suggests that 5000 most popular passwords are used by almost 10\% of the users ($7 \ million$ users) . But there is also a long list of infrequent passwords ($count\le5$) which constitutes a major fraction of the password database. For instance, 53\% ($16.98\ million$) of the passwords in Rockyou database and 56\% ($38.59 \ million$) of the passwords in Yahoo database have a frequency count of  less than five. These infrequent passwords constitute the heavy tail of the distribution and it is difficult to estimate their probabilities.

Based upon these observations, we provide another instance of a partition attacker which we refer to as {\em hybrid bin attacker}. This attacker creates unit-sized partitions for every popular password found in the breached databases. The density of the unit-sized partition containing a single password is equal to the password frequency count. For illustration purpose, assume that the hybrid bin attacker uses Rockyou dataset for the training purpose. The most popular password in the Rockyou database is ``123456" with 290,729 occurrences. The attacker constructs a unit-sized partition for this password and sets its density to 290,729. Table~\ref{tab:popular} shows that the password ``123456" is the popular choice in the remaining databases as well which demonstrates the advantage of creating unit-sized partitions with popular passwords.

\begin{table}[h]
\centering
\tiny
\caption{Top 5 most popular passwords of 6 datasets.}~\label{tab:popular}
\begin{tabular}{|c|c|c|c|c|c|c|c|} \hline
Rank & Rockyou & Gmail & Mail.ru & Yandex & Yahoo & Gawker & Phpbb\\ \hline
1 & 123456       &	123456	& qwerty       &	123456	& 123456	&123456	& 123456 \\
2 & 12345         &	password	& 123456      &	123456789	& password	&password	& password \\
3 & 123456789 &	123456789	& qwertyuiop &	111111	& welcome	&12345678	& phpbb \\
4 & password    &    12345		& qwe123      &	qwerty	& ninja	&lifehack	& qwerty \\
5 & iloveyou	    &     qwerty	& qweqwe	&	1234567890	& abc123	&qwerty	& 12345 \\
\hline
Top 5(\%) & 1.70\% &     1.71\%	& 5.56\%	&	6.07\%	& 0.78\%	& 0.68\%	& 2.18\% \\
\hline
\end{tabular}
\end{table}

After creating unit-sized partitions, the hybrid bin attacker creates bin partitions and estimates their densities from the available password data. By $Theorem\ 1$, the attacker then sorts the resulting partitions in decreasing order of density and attacks the target database. In this way, the hybrid bin attacker can first target the popular passwords by creating a unit-sized partitions and then target the infrequent passwords by creating bin partitions. By using top 5000 Rockyou passwords which constitutes nearly 20\% ($6.4 \ million$) user accounts of the Rockyou database, the attacker can recover a substantial fraction of passwords in other databases (Table~\ref{tab:top}). A simple way to counter this hybrid bin attacker is to employ the counting scheme as proposed by Schechter {\em et al.}~\cite{Schechter} which bans the password after it reaches a predefined frequency threshold. With this countermeasure in place, there are no popular passwords to exploit and the efficiency of the hybrid bin attacker is reduced to that of the bin attacker. However, this countermeasure requires abandonment of password-specific salts, as it is not possible to count the occurrences of a password if every password has a unique salt. 

\begin{table}[h]
\centering
\tiny
\caption{Percentage of passwords cracked in different databases by using top 5000 frequent password of the Rockyou database.}~\label{tab:top}
\begin{tabular}{|c|c|c|c|c|c|} \hline
Gmail & Mail.ru & Yandex & Yahoo & Gawker & Phpbb\\ \hline
530,387(10.76\%) & 417,669(8.95\%) & 175,354(13.90\%) & 47,495(10.72\%) & 75,871(7.00\%) & 36,713(14.37\%)\\ 
\hline
\end{tabular}
\end{table}

\subsection{Salting and Slow Hashes}

It is possible to slow down the offline attacker by using password-specific salts and employing iterative algorithms such as PBKDF2~\cite{Kaliski:2000:PPC:RFC2898}. Now, we discuss the effect of such techniques on the efficiency of both the bin attacker and the hybrid bin attacker.
\\
\textbf{Salting.} The use of password-specific long random salts serves two purposes~\cite{guide}, makes it infeasible to use a rainbow table~\cite{oechslin} and makes it more time-consuming to crack a large list of passwords.
\begin{enumerate}
\item The use of salts prevents pre-computation attacks such as rainbow tables. If passwords are just hashed then the attacker can pre-compute hashes for commonly used millions of passwords well-in advance before acquiring the password database. After stealing the password database, the password hashes can be reversed immediately into plaintext by performing lookups in the rainbow table. Performing lookup is considerably faster than computing the hash function which speeds up the cracking process substantially.  On the other hand, if the password database is salted, then the rainbow table would have to contain ``$guess||salt$" pre-hashed. The value of salt is not known to the attacker until the password database is stolen. In this case, the attacker can pre-compute rainbow table for every value of the salt. However, every bit of salt doubles the storage requirement and if the salt is sufficiently long (128 bits) then this pre-computation is infeasible. As the goal of the salt is to prevent pre-generated databases from being created, it is stored in plaintext. Hence, once the salted password database is compromised, the attacker can learn the salt and start the password cracking process. 
{\em The use of long random salts forces the attacker to crack the hashes after acquiring the password database, instead of being able to just look them all up in a rainbow table. This is how our attacker is modelled, we do not assume that a partition attacker pre-computes the hashes in the rainbow table.} 

\item The use of password-specific salts also conceals the frequency distribution of passwords. If multiple users have same passwords then applying the hash function results in same hashes. The offline attacker can simply count the frequency of each password hash in the database and then use the computing power to crack only unique hashes. If each password has a unique salt then hashes will be different even for the same password. Because of the password-specific salt the attacker cannot simply use the counting technique to find out if two passwords in the database are same or not. The attacker has to compare a guess against each password entry. If there are $\phi$ entries in the salted password database, verifying a single guess in this scenario requires $O(\phi)$ comparisons. Thus, the use of long random salts slows down the attacker by a factor of $\phi$, where $\phi$ is the number of users in the system. Previously, if an attacker could verify $2^G$ guesses against the password database, after using the long random salts, the number of guesses is reduced to $2^G/2^{log(\phi)}=2^{G-log(\phi)}$. 

In section 3, we showed that the bin attacker with the computing power of $2^{56}$ can break more than 90\% of passwords in 6 different password databases (Figure 6a). As most of the test datasets have over a million entries $\phi \approx 2^{20}$ and assuming that these passwords were salted with long random hashes, the bin attacker with the computing power of $2^{56}$ can now verify $2^{G-log(\phi)} = 2^{56-20} = 2^{36}$ guesses. As shown in Figure 6a, the bin attacker with $2^{36}$ guesses can still compromise nearly 50\% of the passwords in the most target password databases which is a significant proportion. {\em Thus, with the use of random salts, the number of guesses that the bin attacker can generate have reduced, but the bin attacker can still compromise significant portion of the password database.} 

\item Further, the use of password-specific salts refrains the use of count-min sketch-based countermeasure proposed by Schechter {\em et al}~\cite{Schechter}. As a result, few passwords become very popular which results in Zipf's law. As explained earlier, the hybrid bin attacker can compromise a substantial portion of the password database using just a few thousand popular guesses. {\em Although, the use of password-specific salts conceals the frequency distribution of passwords and slows down the offline attacker, it does not prevent the Zipf's law in human-generated passwords which can be exploited by the hybrid bin attacker.}  

\item The purpose of salting is only to slow down the offline attacker who possesses the password database. It does not protect against the online attacker who mounts the online guessing attacks on a remote website. The online attacker exploits the fact that human-generated passwords are highly biased and as a result some passwords are extremely popular. For instance, the count of most popular password in Yahoo database~\cite{Yahoo} is 753,217 (1\%). {\em Thus, online attacker can compromise nearly 1\% of the total accounts by merely trying a single guess on a website and use of salt in this attack does not matter.}

\item Further, salting does not slow down the offline attacker who is only after one or few accounts. It helps only if the attacker wants to break many passwords from the password database containing large number of entries. The attacker who wants to compromise only one important account can generate all $2^{56}$ guesses as for this particular attack $\phi=1$. For instance, if the attacker is after, say, admin's account or the account of an influential person such as a celebrity then the attacker can learn the value of the salt and compute $hash(guess||salt)$ for all $2^{56}$ guesses. {\em Thus, the use of a salt does not slow down the targeted attack. The impact of salt is high only if  the number of users $\phi$ is large and the objective of the offline attacker is to break as many passwords as possible.}

\item Recently, researchers~\cite{cryptoeprint:2016:153} emphasized the importance of releasing password frequency list of a website for deciding the various policy parameters. To enable this, they presented a differential privacy based mechanism for releasing the perturbed password frequency lists with rigorous security, efficiency, and distortion guarantees. Now, the knowledge of password frequency distribution can help organizations in setting various policy parameters, for instance the number of unsuccessful attempts $k$ before the account is locked. A smaller value of $k$ can decrease the usability of the authentication experience, while selecting a larger value of $k$ can reduce security. The empirical data from password frequency list could help organizations to make a more informed decision when considering the trade-off between security and usability. {\em Since salting conceals the frequency distribution of passwords, this proposed mechanism~\cite{cryptoeprint:2016:153} requires abandonment of password-specific salts.}
\end{enumerate}
\textbf{Slow Hashes.} Conventional hash functions such as MD5, SHA1, SHA2 can be computed quickly by employing custom made GPU clusters. To make the task of guessing a single password harder (costlier), one can use iterative algorithms such as PBKDF2~\cite{Kaliski:2000:PPC:RFC2898}. By choosing appropriate number of iterations, PBKDF2 can limit the number of guesses to 40,000 guesses per second~\cite{slow}. If passwords are hashed with unique salts, it further slows down the offline attacker. However, if Zipf's law is prevalent in the password data then the hybrid bin attacker requires just a few thousand guesses to recover a substantial fraction of the target database (Table~\ref{tab:top}). On the other hand, if passwords are not salted (to prevent the Zipf's law phenomena) then the bin attacker can compute $2^{36.59}$ guesses using a machine capable of generating 40,000 guesses per second for a month. With $2^{36.59}$ guesses, the attacker can still compromise nearly 50\% passwords in the target password database of size $\phi=2^{20}$ (Figure 6a) which is a significant proportion. {\em The use of PBKDF2 has certainly slowed down the bin attacker but it did not make the attacker completely ineffective.}

\subsection{Long Passwords}
Recent research~\cite{Shay2,Shay2014} suggests that longer passwords are more usable and secure. However, if longer passwords are composed mainly of shorter popular passwords, the password strength will be reduced. For instance if the longer password has ``123456" or ``password" as a substring then its strength is reduced. We used the top 1000 lowercase passwords and analysed  longer lowercase $L^{l}$ passwords, where $l >=12$, in every breached database. The results in Table~\ref{tab:long} clearly indicate the presence of popular substrings in the longer passwords which can be exploited by the offline attacker.
\begin{table}[h]
\centering
\scriptsize
\caption{Percentage of $L^l$ passwords($l >=12$) containing popular substrings.}~\label{tab:long}
\begin{tabular}{|c|c|c|} \hline
Database &  Longer passwords  &  Percentage\\ \hline
Rockyou	      &690,823		&44.75\\
Gmail	      &73,551		&34.05\\
Mail.ru    &74,516		&35.21\\
Yandex   &16,967		&27.65\\
Yahoo   	      &6,546		&34.03\\
Phpbb	                &2,340		&44.48\\
\hline\end{tabular}
\end{table}
\\
\subsection{Bins Utilization}
Dividing the search space into {\em bins} enables us to gauge the utilized search space and therefore determine the effort required for the bin attacker to break the password database in the event of a breach. If the densities of password bins are highly skewed then even an attacker with a small computing power can recover a major fraction of the password database. Most websites attempt to counter offline attacks by forcing users to create passwords from a larger search space, accepting only those passwords derived using a large alphabet set and satisfying minimum length requirement. However, users can respond to this requirement by choosing passwords from bins that requires least effort to recall which again results in few denser bins and can be easily exploited by the bin attacker. Figure~\ref{fig:binutil} depicts the huge gap between the available number of bins and actually utilized bins as a function of password length $l$. The number of available bin increases exponentially with the password length $l$, however the number of utilized bins increases only linearly. As a result the search space is underutilized which can be exploited by the bin attacker.

\begin{figure}[h]
\centering
\includegraphics[scale=0.7]{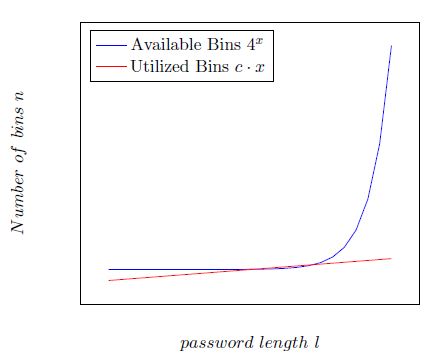}
\caption{Number of available bins vs number of utilized bins as a function of password length $l$.}~\label{fig:binutil}
\end{figure}

Enforcing the same composition policy and length requirement on all users can again lead to the small number of denser bins. Therefore, the system should guide users in the password creation process and we believe that using {\em bin explorer} schemes is a step ahead in that direction. Earlier (in Table~\ref{tab:dataset}) we saw that the number of bins used in the Rockyou database was 140,401. Thus, the capability of users in remembering passwords from different bins should not be undermined. If $32\ million$ Rockyou users would have used even $100,000 < 140,401$ bins uniformly, the attacker would have required huge effort. {\em The uniform density makes the task of offline attacker exponentially difficult due to the large number of bins as well their huge average size, both exponential in terms of password length.}
\section{Conclusion}
The existing probabilistic measures such as entropy and partial guessing metric which are used to measure the offline guessing resistance of password databases are based upon stronger assumptions. These measures require the knowledge of password probabilities which are not always easy to learn. For instance, it is difficult to estimate the probabilities of infrequent passwords in a heavy-tailed distribution. The {\em partition attack model} proposed in this paper does not assume the knowledge of password probabilities. The {\em partition attacker} uses information from previous breaches and, divides the search space into partitions and learns the partition densities. The partitioning approach enables us to measure the utilized search space more effectively thereby providing a more accurate estimate of the password guessing effort in the event of a database breach. 

The {\em partition attack model} is more general as different partitioning strategies represents different instances of a partition attacker. On one extreme, there is an all-powerful attacker, equipped with the knowledge of password probabilities, who divides the search space into unit sized partitions and explores them in decreasing order of density. On the other extreme there is a brute-force attacker who has zero knowledge and treats the entire search space as a single partition. The bin attacker, pre-terminal density attacker and hybrid bin attacker lie somewhere in the middle of attack spectrum since these attacks are more efficient than brute-force attacks but less efficient than probabilistic attacks.

The bin densities in the real-world password databases are highly skewed and even a weak attacker can cause sufficient damage to the system. The pre-terminal partitions are more granular than bins and therefore, the pre-terminal density attacker is more efficient than the bin attacker. The hybrid bin attacker on the other hand create more granular partitions than pre-terminal density attacker and is therefore more effective than pre-terminal partitioning.

The {\em partition attacker} can be thwarted by ensuring that all partitions are equally dense. We proposed a {\em bin explorer system} which first determines the minimum password length by considering the number of users in the system and then uses the random bin distribution strategy to achieve the uniform bin densities, thereby making the task of the bin attacker exponentially difficult. Since the passwords are created using the system assigned bins, the resulting password distribution does not obey Zipf's law and consequently the hybrid bin attacker and pre-terminal density attacker are also countered. 

The usability study results suggest that passwords created using the system assigned bins are not difficult to remember. We found that users can benefit by storing the assigned bins and employing them as hint later during the password recall. Further, users also preferred random bins for protecting their critical accounts. 
\\
\textbf{Future Work.} In this paper, we demonstrated that the partition attacker gets a huge benefit by exploring denser partitions. However, this greedy attacker is unobservant since it does not learn any information while exploring the partitions. It would be interesting to model a dynamic attacker that learns information during the attack and employs it to target the remaining passwords. Also, encouraged by the usability results of the iterative scheme, we intend to explore more schemes to help users in choosing passwords from the system assigned bins.

\newcommand{\BIBdecl}{\setlength{\itemsep}{0 em}}
\bibliographystyle{IEEEtranS}
\bibliography{COSE_ref}

\end{document}